\documentclass[11pt, onecolumn, draftcls]{IEEEtran}
\usepackage{amsmath}
\usepackage{amssymb}
\usepackage{amsthm}
\usepackage{fullpage}
\usepackage{graphicx}
\begin{document}
\newtheorem{theorem}{Theorem}
\newtheorem{lemma}{Lemma}
\newtheorem{mydef}{Definition}

\title{Capacity Bounds for Relay Channels with Inter-symbol Interference and Colored Gaussian Noise}
\author{Chiranjib Choudhuri and Urbashi Mitra}
\maketitle
\long\def\symbolfootnote[#1]#2{\begingroup%
\def\thefootnote{\fnsymbol{footnote}}\footnote[#1]{#2}\endgroup}
\symbolfootnote[0]{Chiranjib Choudhuri (cchoudhu@usc.edu) and Urbashi Mitra (ubli@usc.edu) are with the Ming Hsieh Department of Electrical Engineering, University of Southern California, University Park, Los Angeles 90089.}
\long\def\symbolfootnote[#1]#2{\begingroup%
\def\thefootnote{\fnsymbol{footnote}}\footnote[#1]{#2}\endgroup}
\symbolfootnote[0]{This research has been funded in part by the following grants and organizations: NSF OCE-0520324,  NSF CNS-0722073, NSF CNS-0821750 (MRI), and the University of Southern California's Provost Office.}
\long\def\symbolfootnote[#1]#2{\begingroup%
\def\thefootnote{\fnsymbol{footnote}}\footnote[#1]{#2}\endgroup}
\symbolfootnote[0]{Parts of this work have been previously presented in \cite{Chiru09} and \cite{ChiruWuwnet}.}
\thispagestyle{empty}
\setcounter{page}{0}
\begin{abstract}
The capacity of a relay channel with inter-symbol interference (ISI) and additive colored Gaussian noise is examined under an input power constraint.  Prior results are used to show that the capacity of this channel can be computed by examining the circular degraded relay channel in the limit of infinite block length. The current work provides single letter expressions for the achievable rates with decode-and-forward (DF) and compress-and-forward  (CF) processing employed at the relay.  Additionally, the cut-set bound for the relay channel is generalized for the ISI/colored Gaussian noise scenario.  All results hinge on showing the optimality of the decomposition of the relay channel with ISI/colored Gaussian noise into an equivalent collection of coupled parallel, scalar, memoryless relay channels. The region of optimality of the DF and CF achievable rates are also discussed. Optimal power allocation strategies are also discussed for the two lower bounds and the cut-set upper bound. As the maximizing power allocations for DF and CF appear to be intractable, the desired cost functions are modified and then optimized. The resulting rates are illustrated through the computation of numerical examples.
\end{abstract}

\newpage
\section{Introduction}
\par The relay channel was introduced by van der Meulen \cite{Meulen68, Meulen71, Meulen77} and extensively studied since that time. A significant set of contributions to the analysis of such channels was provided by Cover and El Gamal in \cite{Cover79}, wherein capacity achieving coding strategies were provided for degraded, reversely degraded and feedback relay channels. The bulk of the research on relay channels has focused on memoryless channels either with or without feedback (see {\em e.g.} \cite{Zhang:1988, Reznik:2004,Kramer:2005,ashu_it}). In the current work, we determine the achievable rates and an upper bound on the capacity of a three node simple relay channel with intersymbol interference (ISI) and additive colored Gaussian noise. Such channels are of interest as most of the wireless standards are bandlimited in nature; further, underwater acoustic channels exhibit both ISI and colored noise (see \cite{milica2006, Vajapeyam07, Cecilia06}). It is needless to say that the problem is challenging as even the capacity of the simple memoryless relay channel is a long standing open problem with solutions for scenarios under very specific conditions (see \cite{Kramer:2005},\cite{Young}).

\par In this paper, we discuss two important coding strategies at the relay (a) the Decode-and-Forward (DF) protocol, and (b) the Compress-and-Forward (CF) protocol, and derive the corresponding achievable rates. In addition, we generalize the cut-set bound for the converse to our scenario of interest.  As CF and DF have differing regimes in which they offer the best rate \cite{Kramer:2005} for memoryless channels, it is of interest to investigate both coding strategies for the relay channel with finite memory and colored Gaussian noise as we do herein.

\par Important prior work on this problem includes \cite{Andrea01}, which provided the link between circular multi-terminal networks with ISI and linear ones.  A single-letter expression for the two-user broadcast channel is given in \cite{Andrea01}; however, the computational methods used therein do not directly apply to our case due to the presence of the multihop link. In fact, the challenge for the three node relay network stems from the intermediate processing at the relay node.  While the defining expressions for the capacity of the relay channel with finite memory\footnote{Classically, capacity computations in the presence of memory require examination of the entire signal due to the memory.} are provided in \cite{Ninoslav08}, a single-letter expression is not provided.  As in \cite{Andrea01,Hirt88,Cheng93}, we employ the discrete Fourier transform (DFT) to decompose our circular relay channel into a collection of parallel scalar relay channels.  A key consequence of our work is that the parallel decomposition is optimal for the computation of both the DF and CF achievable rates and thus permuting channels at the relay \cite{wenyi} cannot improve the bounds.  The resulting parallel relays are coupled via the power constraint for the DF case and via both a power and rate constraint for the CF case which affect the optimal power allocation strategies.

\par Allocation of resources (power, bandwidth, bit rates) in the context of specific coding schemes are long studied in communication community, focusing mainly on one hop multi-user channels (see \cite{Andrea01, Cheng93, Hanly}). In this paper we have studied the optimal power policy that will achieve the modified capacity bounds and also shown that in some of the cases simple water-filling type of power allocation at the nodes are optimal.

\par The remainder of the paper is organized as follows. In Section II, we introduce the different channel models and we establish the equivalence between different channel models, while Section III computes the bounds on the capacity of $n$-block channel and the condition of optimality of the different achievable rates. In Section IV we examine the bounds in the limit of infinite block length. In Section V, we derive the power allocation strategies for our various bounds. Section VI provides a few illustrative examples. Finally, in Section VII, we present the conclusions of the work.

\begin{figure}
	\centering
		\includegraphics[width=0.50\textwidth]{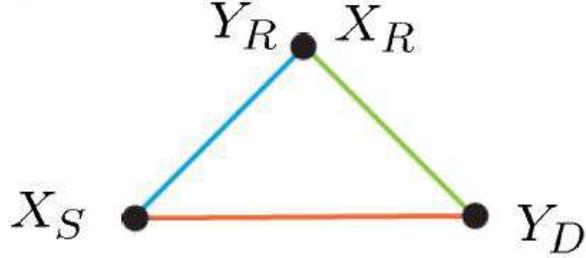}
	\caption{Channel model of a single-relay channel}
	\label{fig:new_relay}
\end{figure}

\section{Channel Models and Capacity Relationships}
\par In this section, we introduce our channel model and the related circular Gaussian relay channel building on the formulations of \cite{Andrea01}.
We consider the capacity of the discrete-time relay channel model as shown in Fig. 1. The signal transmitted by the source and relay are given by \{$x_{Sk}$\} and \{$x_{Rk}$\}, respectively.  The stationary, additive Gaussian noise processes at the relay and destination, denoted by \{$v_{Rk}$\} and \{$v_{Dk}$\}, have  zero-mean, and autocorrelation functions $R_R[i]$ and $R_D[i]$, respectively, of finite support $i_{max}$. Let $\{h_{qi}\}_{i=0}^m,\;\;q\in\{SR,SD,RD\}$ denote the channel impulse responses (CIRs) of the three links, with common memory length $m$. Without loss of generality, we only consider the case,  $m\geq i_{max}$\footnote{If $m < i_{max}$, CIRs can be zero padded to make them equal.}.  The output sequences at the relay and destination are,  \{$y_{Rk}$\} and \{$y_{Dk}$\}, respectively, with,
\begin{eqnarray}\label{eq:in-out_lin}
y_{Rk} & = & \sum_{i=0}^mh_{SRi}x_{S(k-i)}+v_{Rk},\nonumber\\
y_{Dk} & = & \sum_{i=0}^m(h_{SDi}x_{S(k-i)}+h_{RDi}x_{R(k-i)})+v_{Dk}.
\end{eqnarray}
The following power constraints are assumed for the source and relay signals, for all $n$,
\begin{eqnarray}\label{eq:power_cons}
\frac{1}{n}\sum_{k=1}^nE[x_{qk}^2]=\frac{1}{n}\beta_q^n\leq P_q, \;\; q \in \{S,R\}.
\end{eqnarray}
For a given $m$, this channel is called the \textit{linear Gaussian relay channel} (LGRC) with finite memory $m$ (see \cite{Andrea01}) as the output is a linear convolution of the input codeword with the channel impulse response. Clearly the channels have ISI since the channel output at time $k$ depends on the input symbols $\{x_q(k)\}_{q\in\{S,R\}}$ at time $k$ as well as previous input symbols $\{x_q(i)\}_{q\in\{S,R\}}, i<k$. In addition, the noise samples $\{v_q(k)\}_{q\in\{R,D\}}$ at time $k$ is correlated with noise samples $\{v_q(i)\}_{q\in\{R,D\}}$ at times $i<k$.

\par We now define the $n$-block \textit{circular Gaussian relay channel} ($n$-CGRC) for $n>m$, a $n$-block memoryless channel obtained by modifying the LGRC with memory $m$. This channel model will play a pivotal role in capacity computation as seen in the sequel. Specifically, the $n$-CGRC over each $n$-block has input vectors $\{x_{Sk}\}_{k=1}^n, \{x_{Rk}\}_{k=1}^n$ which produce output vectors $\{{y}_{Rk}^c\}_{k=1}^n$ at the relay and $\{{y}_{Dk}^c\}_{k=1}^n$ at the destination with,
\begin{eqnarray}\label{eq:in-out_circ}
y_{R}^{nc} & = & \textbf{H}_{SR}^cx_{S}^n+v_{R}^{nc},\nonumber\\
y_{D}^{nc} & = & \textbf{H}_{SD}^cx_{S}^n+\textbf{H}_{RD}^cx_{R}^n+v_{D}^{nc},
\end{eqnarray}
where $\textbf{H}^c$ is the circulant channel matrix, whose first row is defined as,
\begin{eqnarray*}
\textbf{H}_q^c(1,:) & = & [h_{q0}, 0, \cdots, 0, h_{qm}, \cdots, h_{q2}, h_{q1}],\;\; q\in \{SR,SD,RD\},
\end{eqnarray*}
and each subsequent row is a single cyclic shift to the right.
\par The only difference with the linear channel model is that the channel output is the circular convolution of the input codeword with the channel impulse response instead of a linear one. The circular noise processes over each $n$-block $\{{v}_{qk}^c\}_{k=1}^n,\;\;q\in\{R,D\}$ have periodic autocorrelation function and can be found in \cite{Andrea01}. Noise samples from different $n$-blocks are independent since the channel is $n$-block memoryless. The $n$-CGRC inherits the LGRC's power constraints. 
\par We next define key notation which will be used throughout the paper.  We let
$\Sigma_q = E[x_q^n(x_q^n)^{\dag}],$ and $\textbf{N}_q = E[v_q^{nc}(v_q^{nc})^{\dag}], \; \; q \in \{S,R\}$ be the source and relay input correlation matrices and noise correlation matrices, respectively. We shall repeatedly make use of the fact that circulant matrices can be diagonalized by the Discrete Fourier Transform (DFT) matrix, which we denote as \textbf{F}. Thus, $X_q^n=\textbf{F}x_q^n, \; \; q \in \{S,R\}$ is the DFT of the input signal $x_q^n$ with $\Psi_{q} =E[X_q^n(X_q^n)^{\dag}],\; \; q \in \{S,R\}$. Also for a matrix $\textbf{A}$, $|\textbf{A}|$ denotes the absolute value of the determinant of $\textbf{A}$. Additionally, the following diagonal matrices are defined:
\begin{eqnarray}
\textbf{D}_l & = & \textbf{F}\textbf{H}_l^c\textbf{F}^{\dag}, \;\; l\in \{SR,SD,RD\}, \label{eq:dl}\\
\textbf{C}_q & = & \textbf{F}\textbf{N}_q\textbf{F}^{\dag}, \;\; q\in \{R,D\}. \label{eq:cq}
\end{eqnarray} 
Finally, we define $x^n=x_S^n|x_R^n$. Then, $\Sigma=E[x^n(x^n)^{\dag}]=\Sigma_S - \Sigma_{SR}\Sigma_{R}^{-1}\Sigma_{SR}^{\dag}$ and $\Psi = \textbf{F}\Sigma \textbf{F}^{\dag} = \Psi_S - \Psi_{SR}\Psi_{R}^{-1}\Psi_{SR}^{\dag}$.

\par Direct computation of the capacity of the n-block LGRC is challenged by the presence of inter-block interference.
In \cite{Andrea01}, it is shown that if we extend the definitions of the $n$-LGRC, the $n$-CGRC to a synchronous Gaussian multi-terminal channel, then the capacity region of the two multi-terminal channels is the same in the limit as $n$ goes to infinity. As the relay channel is a special case of a synchronous multi-terminal channel, we get our desired result. Thus, the capacity $C$ of the LGRC can be computed as the limit of the $n$-CGRC, as $n$ grows to infinity.  In the sequel, we derive single-letter expressions for a variety of rate bounds:  achievable rates for DF and CF, and a generalization of the cut-set bound for the $n$-LGRC by exploiting this equivalence. 

\section{Bounds on the Capacity of $n$-CGRC}

\subsection{Achievable Rate: Decode-and-Forward}

\par Since the $n$-CGRC defined in Eqn. (\ref{eq:in-out_circ}) is an $n$-block memoryless relay channel, its achievable rate under DF coding strategy follows directly from \cite{Cover79} if we replace $(X,X_1,Y_1,Y)$ by $(x_S^n,x_R^n,y_R^{nc},y_D^{nc})$. The DF achievable rate is thus given by
\begin{eqnarray*}
C_{nDF}^c(P_S,P_R)=\sup_{p(x_S^n,x_R^n)} \frac{1}{n} \min\{I(x_S^n;y_R^{nc}|x_R^n),I(x_S^n,x_R^n;y_D^{nc})\},
\end{eqnarray*}
satisfying the power constraints given by (\ref{eq:power_cons}).
\par To simplify notation, we define the following power constraint set
\begin{eqnarray*}
{\cal P}^D & = & \left \{\alpha(\cdot),P_S(\cdot), P_R(\cdot): 0\leq\alpha(\omega_i)\leq 1, \frac{1}{n}\sum_{i=1}^{n}P_{S}(\omega_i)\leq P_S, \frac{1}{n}\sum_{i=1}^{n}P_R(\omega_i)\leq P_R \right\}.
\end{eqnarray*}
\begin{theorem} \label{thm:df}
The achievable rate for a $n$-block CGRC with finite memory $m$, where the relay employs DF, is given by
\begin{eqnarray}
C_{nDF}^c(P_S,P_R)& = & \max_{{\cal P}^D} \min \{C_{1nDF}^c, C_{2nDF}^c\}, \nonumber\\
\mbox{where,}\; \:
C_{1nDF}^c & = &\frac{1}{2n}\sum_{i=1}^{n}C \left(\frac{\alpha(\omega_i)\left|H_{SR}(\omega_i)\right|^2P_S(\omega_i)}{N_{R}(\omega_i)}\right), \nonumber\\
C_{2nDF}^c & = & \frac{1}{2n}\sum_{i=1}^{n}C \left(\frac{P(\omega_i)}{N_{D}(\omega_i)} \right), \nonumber\\
 H_q(\omega_i)& = & \textbf{D}_{qii},q\in \{SR,SD,RD\}, \label{eq:channel_i}\\
N_q(\omega_i)& = & \textbf{C}_{qii},q\in \{R,D\},\label{eq:noise_i}
\end{eqnarray}
\noindent and
\begin{eqnarray*}
P(\omega_i) & = & \left|H_{SD}(\omega_i)\right|^2P_S(\omega_i)+\left|H_{RD}(\omega_i)\right|^2P_R(\omega_i)+ 2\sqrt{\bar{\alpha}(\omega_i)\left|H_{SD}(\omega_i)H_{RD}(\omega_i)\right|^2P_S(\omega_i)P_R(\omega_i)}.
\end{eqnarray*}
Additionally, $C(x)= \log(1+x)$.  The $ H_q(\omega_i)$ are the channel components for frequency bin $i$ and link $q$; the $N_q(\omega_i)$ are similarly defined noise components.  The $P_S(\omega_i)$ and $P_R(\omega_i)$ are the powers allocated by the source and the relay, respectively, for $i$'th component of the channel and $0 \leq \alpha(\omega_i) \leq 1$ is the cross-correlation between the input signals as defined in \cite{Cover79} and $\bar{\alpha}(\omega_i)=1-\alpha(\omega_i)$.
\end{theorem}
Before proving the theorem, we introduce two key lemmas.  We first require a property of the maximizing input probability distribution  from \cite{Ninoslav08}.
\begin{lemma} \label{lemma:ninoslav}
\cite{Ninoslav08} The capacity of the degraded relay channel with finite memory of length m is
\begin{eqnarray*}
C_{DF} & = & \lim_{n\rightarrow\infty}\sup_q \frac{1}{n} \min\{I(x_S^n;y_R^n|x_R^n),I(x_S^n,x_R^n;y_D^n)\},
\end{eqnarray*}
where the maximization is taken over the input distribution $q = \prod_{i=1}^n p(x_{Si}|x_{Si-1}, x_{Ri})p(x_{Ri}|x_{Ri-1})$.
\end{lemma}
\noindent
Lemma~\ref{lemma:ninoslav} implies that the process $x_R$ is allowed to evolve without any dependence on the process $x_S$, while the process $x_S$ may be causally dependent on $x_R$.  The key to showing Theorem~\ref{thm:df} is proving that the DFT decomposition is optimal for DF.  To this end, we must show that a certain correlation structure holds for the source and relay signals, $x_R$ and $x_S$.

\begin{lemma} \label{lemma:diag}
Given $\Psi_R$ diagonal, for jointly Gaussian input $(X_S^n,X_R^n)$ of the form of Lemma 1, the matrices $\Psi, \textbf{D}$ will be diagonal if and only if $\Psi_S$  and $\Psi_{SR}$ are diagonal matrices, where
\begin{eqnarray}\label{eq:D}
\textbf{D}& = & \textbf{D}_{SD}\Psi_S \textbf{D}_{SD}^{\dag}+\textbf{D}_{RD}\Psi_R \textbf{D}_{RD}^{\dag}+2Re(\textbf{D}_{SD}\Psi_{SR}\textbf{D}_{RD}^{\dag})+\textbf{C}_D.
\end{eqnarray}
\end{lemma}
\begin{proof} For a given diagonal $\Psi_R$, if $\Psi_S$ and $\Psi_{SR}$ are diagonal matrices, it is easy to see that both $\Psi$ and $\textbf{D}$ are diagonal. To show that diagonal $\Psi$ and $\textbf{D}$ implies diagonal $\Psi_S$ and $\Psi_{SR}$ for a given diagonal $\Psi_R$, we proceed as follows.
\par As the input vectors are jointly multivariate Gaussian we can decompose our source, in the DFT domain, into
\begin{eqnarray}\label{eq:Lin_reg}
X_S^n & = & \textbf{V}X_R^n + \textbf{W}X_{S0}^n,
\end{eqnarray}
where, $X_S^n,X_R^n$ are the DFTs of the input symbols $x_S^n,x_R^n$, $\textbf{V}$ and $\textbf{W}$ are general $n\times n$ matrices and $X_{S0}^n$ are a set of $n$ independent Gaussian random variables, also independent of $X_R^n$.  From Lemma~\ref{lemma:ninoslav}, it is sufficient to consider only lower triangular matrices $\textbf{V}$ and $\textbf{W}$.  Substituting the value of $X_S^n$ from Eqn. (\ref{eq:Lin_reg}) in $\Psi$, we get
\begin{eqnarray*}
\Psi & = & \textbf{V}\Psi_{R}\textbf{V}^{\dag} + \textbf{W}\Psi_{S0}\textbf{W}^{\dag} - \textbf{V}\Psi_{R}\Psi_{R}^{-1}\Psi_{R}\textbf{V}^{\dag}
= \textbf{W}\Psi_{S0}\textbf{W}^{\dag},
\end{eqnarray*}
where $\Psi_{S0}$ is the covariance matrix of $X_{S0}^n$, and is diagonal by construction. As the product of a non-singular lower and upper triangular matrix is diagonal if and only if they themselves are diagonal, $\textbf{W}$ must be diagonal for $\Psi$ to be diagonal, as $\textbf{W}$ is lower triangular and $\Psi_{S0}\textbf{W}^{\dag}$ is upper triangular.
\par To show the related result for $\textbf{D}$, we first assume, without loss of generality, that the channel diagonal matrices are the identity.  Then, substituting Eqn. (\ref{eq:D}), we have,
\begin{eqnarray*}
\textbf{D} & = & (\textbf{V}+\textbf{I})\Psi_R(\textbf{V}+\textbf{I})^{\dag}+\textbf{W}\Psi_{S0}\textbf{W}^{\dag}.
\end{eqnarray*}
As argued above for $\textbf{W}$, only a diagonal $\textbf{V}$, diagonalizes $(\textbf{V}+\textbf{I})\Psi_R(\textbf{V}+\textbf{I})^{\dag}$ for a diagonal $\Psi_R$. Since $\textbf{V}, \textbf{W}$ and $\Psi_R$ are diagonal, $\Psi_S$ and $\Psi_{SR}$ must be diagonal as well.
\end{proof}

With Lemma~\ref{lemma:ninoslav} and ~\ref{lemma:diag} in hand, we can prove Theorem~\ref{thm:df}.

\begin{proof}
For the Gaussian relay channel
\begin{eqnarray}\label{eq:c1ndf}
C_{1nDF}^c & \equiv & I(x_S^n;y_R^{nc}|x_R^n) = h(y_R^{nc}|x_R^n) - h(y_R^{nc}|x_R^n, x_S^n)\nonumber\\
& = & h(\textbf{H}_{SR}^cx_{S}^n+v_{R}^{nc}|x_R^n) - \frac{1}{2}\log 2\pi e\left|\textbf{N}_R\right|\nonumber\\
& \stackrel{(a)}{\leq} &\frac{1}{2}\log 2\pi e\left|\textbf{H}_{SR}^c\Sigma \textbf{H}_{SR}^{c\dag}+\textbf{N}_R\right|- \frac{1}{2}\log 2\pi e\left|\textbf{N}_R\right|\nonumber\\
& \stackrel{(b)}{\leq} & \frac{1}{2}\log 2\pi e\left|\textbf{D}_{SR}\Psi\textbf{D}_{SR}^{\dag}+\textbf{C}_R\right|- \frac{1}{2}\log 2\pi e\left|\textbf{C}_R\right|,
\end{eqnarray}
where (a) follows from the fact that a Gaussian distribution maximizes the entropy; (b) follows from the fact that both the channel impulse response  $\textbf{H}_l^c$ and noise correlation $\textbf{N}_R$ matrices are circulant and thus can be diagonalized by the DFT matrix $\textbf{F}$ and due to the fact that $\textbf{F}$ is unitary and hence has a unity determinant.  Similarly,
\begin{eqnarray}
\lefteqn{C_{2nDF}^c \equiv I(x_S^n,x_R^n;y_D^{nc})}\nonumber\\
& \leq &  \frac{1}{2}\log 2\pi e\left|\underbrace{\textbf{D}_{SD}\Psi_S \textbf{D}_{SD}^{\dag}+\textbf{D}_{RD}\Psi_R \textbf{D}_{RD}^{\dag}+2Re(\textbf{D}_{SD}\Psi_{SR}\textbf{D}_{RD}^{\dag})}_{\equiv \textbf{D}}+\textbf{C}_D\right|- \frac{1}{2}\log 2\pi e\left|\textbf{C}_D\right|\nonumber\\
& = & \frac{1}{2}\log 2\pi e \left|\textbf{C}_{D}+\textbf{D}\right|- \frac{1}{2}\log 2\pi e\left|\textbf{C}_D\right| \label{eq:c2ndf}.
\end{eqnarray}
Equality occurs when the input vectors are multivariate Gaussian distributed.
\par  We assume, as in Lemma~\ref{lemma:diag}, that without loss of generality, the diagonal channel matrices are identity matrices. Thus we have

\begin{eqnarray}
C_{2nDF}^c & \stackrel {(a)}{\leq}  & \frac{1}{2}\log\left|\textbf{D}_{1}+\textbf{V}_1\Psi_R\textbf{V}_1^{\dag}\right|\nonumber\\
& \stackrel {(b)}{=} & \frac{1}{2}\log\left|\Psi_R^{-1}+\textbf{V}_1\textbf{D}_{1}^{-1}\textbf{V}_1^{\dag}\right|\left|\Psi_R\right|\left|\textbf{D}_{1}\right|,
\end{eqnarray}
where, $\textbf{V}_1=\textbf{V}+\textbf{I}$ and $\textbf{D}_1=\textbf{C}_D+\textbf{W}\Psi_{S0}\textbf{W}^{\dag}$ and (a) follows from making these substitutions in Eqn.(\ref{eq:c2ndf}) and (b) follows from the Generalized Matrix-Determinant Lemma \cite{David}.  By Hadamard's inequality (see {\em e.g.}\cite{Cover_book}), the determinant of a positive definite matrix is maximized when the matrix is diagonal, thus a diagonal $\Psi_R$ maximizes $C_{2nDF}^c$. As we have that a diagonal $\Psi_R$ is necessary, Lemma~\ref{lemma:diag} provides our final desired result.  The final expression in the theorem results from manipulating Eqns.(\ref{eq:c1ndf}) and (\ref{eq:c2ndf}) and employing the definitions in Eqns. (\ref{eq:channel_i}) and (\ref{eq:noise_i}).
\end{proof}

\begin{figure}
	\centering
		\includegraphics[width=0.50\textwidth]{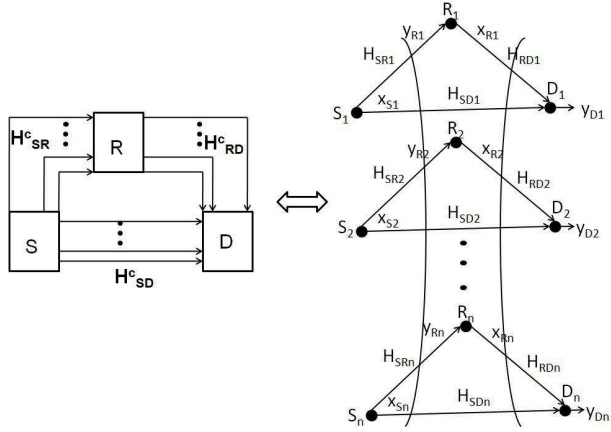}
	\caption{Decomposition of relay with ISI into parallel memoryless relays in frequency domain}
	\label{fig:n_parallel_relay}
\end{figure}

\par The implication of Theorem~\ref{thm:df} is that if we design $x_S$ in the DFT domain, a codeword which is white across sub-channels is optimal. This implies that treating the relay channel as a set of $n$ parallel and independent scalar relay channels is optimum for the computation of the DF rate (as shown in Fig. 2). Thus, other kinds of relay processing such as permuting the channels via channel matching as was done in \cite{wenyi} for the multi-hop channel is sub-optimal, as one cannot exploit potential cooperative gain in a single sub-channel. Observe however, that the input power constraints are {\em coupled} for the $n$ parallel channels.

\subsection{Achievable Rate: Compress-and-Forward with Gaussian Input}
\par In DF, the relay completely decodes the source codeword and then retransmits a related signal of lower rate to the destination. In contrast, in CF, the relay quantizes the received signal and transmits this quantized information to the destination. Since the maximizing input distribution is not known for the Gaussian CF relay channel, we consider inputs with normal pdf.

\begin{theorem}\label{thm:cf}
For $n$-block CGRC defined in the last section, the CF achievable rate with Gaussian inputs is given by,
\begin{eqnarray}\label{eq:CFrate}
C_{nCF}^c & = &  \sup_{\hat{N}_R(\omega_i)\geq 0, 1\leq i \leq n} \sum_{i=1}^n \frac{1}{2n} C\left(P_S(\omega_i)\left(\frac{\left|H_{SD}(\omega_i)\right|^2}{N_D(\omega_i)}+\frac{\left|H_{SR}(\omega_i)\right|^2}{\hat{N}(\omega_i)}\right)\right),
\end{eqnarray}
subject to the input power constraint (\ref{eq:power_cons}) and
\begin{eqnarray}\label{eq:Comp_rate_cons}
\sum_{i=1}^n \log \hat{N}_{R}(\omega_i) & \geq & \sum_{i=1}^n \log\left( \frac{P_S(\omega_i)\left(\left|H_{SR}(\omega_i)\right|^2N_D(\omega_i)+\left|H_{SD}(\omega_i)\right|^2\hat{N}(\omega_i)\right)+\hat{N}(\omega_i)N_D(\omega_i)} {\left|H_{SD}(\omega_i)\right|^2P_S(\omega_i)+\left|H_{RD}(\omega_i)\right|^2P_R(\omega_i)+N_D(\omega_i)}\right),\nonumber\\
\end{eqnarray}
where, $\hat{N}(\omega_i)=N_R(\omega_i)+\hat{N}_{R}(\omega_i)$ and $\hat{N}_{R}(\omega_i)$ is the variance of the quantization noise in the $i$-th sub-band. As in Theorem \ref{thm:df}, the $ H_q(\omega_i)$ are the channel components for frequency bin $i$ and link $q$; the $N_q(\omega_i)$ are similarly defined noise components.  The $P_S(\omega_i)$ and $P_R(\omega_i)$ are the powers allocated by the source and the relay, respectively, for $i$'th component of the channel.
\end{theorem}
\par  Before proving the Theorem \ref{thm:cf}, we introduce one key lemma which states the optimality of decomposed relay for the computation of CF rate.
\begin{lemma}\label{lem:opt_CF}
A diagonal $\Psi_S$  and $\Psi_R$ maximizes $C_{nCF}^c$ in Theorem \ref{thm:cf}.
\end{lemma} 
The proof of this lemma is provided in Appendix A. With the Lemma \ref{lem:opt_CF} in hand, we can now prove Theorem \ref{thm:cf}.   
\begin{proof}
As the $n$-CGRC with memory $m$ is block memoryless, we can extend the CF achievable rate results for scalar relay channels \cite{Cover79} to the vector relay channels to yield the following lower bound on the capacity,

\begin{eqnarray}\label{eq:CF_rate}
C_{nCF}^c  & = & \sup \frac{1}{n}I(x_S^n;y_D^{nc},\hat{y}_R^{nc}|x_R^n)\nonumber\\
& & \mbox{with,}\;\; I(x_R^n;y_D^{nc}) \geq I(y_R^{nc};\hat{y}_R^{nc}|x_R^n,y_D^{nc}),
\end{eqnarray}
where $\hat{y}_R^{nc}$ is the quantized version of $y_R^{nc}$ and the supremum is taken over all joint distributions of the form,
\begin{eqnarray*}
p(x_S^n,x_R^n,y_R^{nc},\hat{y}_R^{nc},y_D^{nc}) & = & p(x_S^n)p(x_R^n)p(y_D^{nc},y_R^{nc}x_S^n,x_R^n)p(\hat{y}_R^{nc}|y_R^{nc},x_R^n).
\end{eqnarray*}

\par  Using Eqn. (\ref{eq:in-out_circ}) and our assumption of Gaussian inputs, we have,
\begin{eqnarray}
& & I(x_S^n;y_D^{nc},\hat{y}_R^{nc}|x_R^n)=  h(y_D^{nc},\hat{y}_R^{nc}|x_R^n)-h(y_D^{nc},\hat{y}_R^{nc}|x_S^n,x_R^n)\nonumber\\
& \stackrel{(a)}{=} & h(\textbf{H}_{SD}^cx_{S}^n+\textbf{H}_{RD}^cx_{R}^n+v_{D}^{nc},y_R^{nc}+\hat{v}_R^{nc}|x_R^n)-h(\textbf{H}_{SD}^cx_{S}^n+\textbf{H}_{RD}^cx_{R}^n+v_{D}^{nc},y_R^{nc}+\hat{v}_R^{nc}|x_S^n,x_R^n)\nonumber\\
& \stackrel{(b)}{=} & \frac{1}{2} \log (2\pi e)^2\left|\left[\begin{array}{cc} \textbf{H}_{SD}^c\Sigma_S\textbf{H}_{SD}^{c\dag}+\textbf{N}_D & \textbf{H}_{SD}^c\Sigma_S\textbf{H}_{SR}^{c\dag}\\\textbf{H}_{SD}^{c\dag}\Sigma_S\textbf{H}_{SR}^c & \textbf{H}_{SR}^c\Sigma_S\textbf{H}_{SR}^{c\dag}+\textbf{N}_R+\hat{\textbf{N}}_R\end{array}\right]\right|
-\frac{1}{2} \log (2\pi e)^2\left|\left[\begin{array}{cc} \textbf{N}_D & 0\\0 & \textbf{N}_R+\hat{\textbf{N}}_R\end{array}\right]\right|\nonumber
\end{eqnarray}
\begin{eqnarray}\label{eq:cf1}
& \stackrel{(c)}{=} & \frac{1}{2} \log (2\pi e)^2\left|\left[\begin{array}{cc} \textbf{D}_{SD}\Psi_S\textbf{D}_{SD}^{\dag}+\textbf{C}_D & \textbf{D}_{SD}\Psi_S\textbf{D}_{SR}^{\dag}\\\textbf{D}_{SD}^{\dag}\Psi_S\textbf{D}_{SR} & \textbf{D}_{SR}\Psi_S\textbf{D}_{SR}^{\dag}+\textbf{C}_R+\hat{\textbf{C}}_R\end{array}\right]\right|-\frac{1}{2} \log (2\pi e)^2\left|\left[\begin{array}{cc} \textbf{C}_D & 0\\0 & \textbf{C}_R+\hat{\textbf{C}}_R\end{array}\right]\right|.\nonumber\\
\end{eqnarray}
Here, (a) follows from the fact that the quantized version of $y_R^{nc}$, $\hat{y}_R^{nc}$ can be written as $y_R^{nc}+\hat{v}_R^{nc}$ from source coding theory (e.g. \cite{Cover_book}), where $\hat{v}_R^{nc}$ is a sequence of independent random variables (whose covariance matrix $\hat{\textbf{N}}_R = \textbf{F}^{\dag}\hat{\textbf{C}}_R\textbf{F}$ will be optimized to maximize the achievable rate), which are also independent of the input vectors and additive noises; (b) follows from the fact that the input vectors are jointly Gaussian, and (c) holds because both the channel matrix $\textbf{H}_l^c$ and noise covariance matrix $\textbf{N}_q$ are circulant by construction; hence they are diagonalized by the unitary matrix $\textbf{F}$. The inequality constraint with the Gaussian inputs can be simplified as follows.
\begin{eqnarray*}
I(x_R^n;y_D^{nc}) & = & \frac{1}{2} \log (2\pi e) \left|\textbf{D}_{SD}\Psi_S\textbf{D}_{SD}^{\dag}+\textbf{D}_{RD}\Psi_R\textbf{D}_{RD}^{\dag}+\textbf{C}_D\right|- h(y_D^{nc}|x_R^n).
\end{eqnarray*}
Similarly, for the RHS of Inequality (\ref{eq:CF_rate}), we have
\begin{eqnarray*}
\lefteqn{I(y_R^{nc};\hat{y}_R^{nc}|x_R^n,y_D^{nc})}\\
& = & h(\hat{y}_R^{nc},y_D^{nc}|x_R^n)-h(y_D^{nc}|x_R^n)-h(\hat{v}_R^{nc})\\
& = & \frac{1}{2}\log (2\pi e)^2\left|\left[\begin{array}{cc} \textbf{D}_{SR}\Psi_S\textbf{D}_{SR}^{\dag}+\textbf{C}_R+\hat{\textbf{C}_R} & \textbf{D}_{SR}\Psi_S\textbf{D}_{SD}^{\dag}\\\textbf{D}_{SD}\Psi_S\textbf{D}_{SR}^{\dag} & \textbf{D}_{SD}\Psi_S\textbf{D}_{SD}^{\dag}+\textbf{C}_D\end{array}\right]\right|-h(y_D^{nc}|x_R^n)-\frac{1}{2}\log (2\pi e) \left|\hat{\textbf{C}}_R\right|.
\end{eqnarray*}
Thus the inequality constraint is given by
\begin{eqnarray}\label{eq:cf_ineq1}
\log \left|\hat{\textbf{C}}_R\right| & \geq & \log \left|\left[\begin{array}{cc} \textbf{D}_{SR}\Psi_S\textbf{D}_{SR}^{\dag}+\textbf{C}_R+\hat{\textbf{C}_R} & \textbf{D}_{SR}\Psi_S\textbf{D}_{SD}^{\dag}\\\textbf{D}_{SD}\Psi_S\textbf{D}_{SR}^{\dag} & \textbf{D}_{SD}\Psi_S\textbf{D}_{SD}^{\dag}+\textbf{C}_D\end{array}\right]\right|\nonumber\\
& &  - \log \left|\textbf{D}_{SD}\Psi_S\textbf{D}_{SD}^{\dag}+\textbf{D}_{RD}\Psi_R\textbf{D}_{RD}^{\dag}+\textbf{C}_D\right|.\label{eq:cf4}
\end{eqnarray}
\par Now using Lemma \ref{lem:opt_CF}, we can not only restrict $\Psi_S$ and $\Psi_R$ to be a diagonal matrices, but as with the computation of the DF achievable rate, a stronger result, which says that the CF achievable rate with Gaussian inputs is maximized when we decompose the network into $n$ parallel, scalar relay channels, can be proved. 

\par Note that a diagonal $\Psi_S$ and $\Psi_R$ alone does not achieve the desired decomposition into parallel relay channels as a diagonal $\Psi_S$ and $\Psi_R$ does not block diagonalize the matrices in Eqn. (\ref{eq:cf1}) and (\ref{eq:cf_ineq1}). The implied statistical independence by diagonal $\Psi_S$ and $\Psi_R$ coupled with a proper orthonormal permutation matrix does achieve our desired result, which we will show next. We define such an orthonormal permutation matrix, $\textbf{P}$, such that,
\begin{eqnarray*}
\textbf{P}_{2j-1,j} & = & 1 \;\; \mbox{for} \;\; j=1,\cdots,n\\
\textbf{P}_{2j-2n,j} & = & 1 \;\; \mbox{for} \;\; j=n+1,\cdots,2n.
\end{eqnarray*}
Employing the defined $\textbf{P}$ and using $\det(\textbf{I}_n + \textbf{P}\textbf{A}\textbf{P}^T)=\det(\textbf{I}_n + \textbf{A})$, we rewrite the Eqn.(\ref{eq:cf1}) as,
\begin{eqnarray}
C_{nCF}^c & = & \lefteqn{\frac{1}{2}\log \left|\textbf{I}_{2n}+ \left[\begin{array}{cc} \textbf{C}_{D}^{-1} & 0\\0 & \left(\textbf{C}_{R}+\hat{\textbf{C}}_R\right)^{-1}\end{array}\right]\left[\begin{array}{cc} \textbf{D}_{SD}\Psi_S\textbf{D}_{SD}^{\dag} & \textbf{D}_{SD}\Psi_S\textbf{D}_{SR}^{\dag}\\\textbf{D}_{SR}\Psi_S\textbf{D}_{SD}^{\dag} & \textbf{D}_{SR}\Psi_S\textbf{D}_{SR}^{\dag}\end{array}\right]\right|}\nonumber\\
& = &  \frac{1}{2}\log \left|\textbf{I}_{2n}+\textbf{P} \left[\begin{array}{cc} \textbf{C}_{D}^{-1} & 0\\0 & \left(\textbf{C}_{R}+\hat{\textbf{C}}_R\right)^{-1}\end{array}\right]\underbrace{\textbf{P}^T\textbf{P}}_{{\bf I}}\left[\begin{array}{cc} \textbf{D}_{SD}\Psi_S\textbf{D}_{SD}^{\dag} & \textbf{D}_{SD}\Psi_S\textbf{D}_{SR}^{\dag}\\\textbf{D}_{SR}\Psi_S\textbf{D}_{SD}^{\dag} & \textbf{D}_{SR}\Psi_S\textbf{D}_{SR}^{\dag}\end{array}\right]\textbf{P}^T\right|.\label{eq:cf3}
\end{eqnarray}
It can be easily shown that,
\begin{eqnarray*}
\textbf{P}\left[\begin{array}{cc} \textbf{C}_D^{-1} & 0\\0 & \left(\textbf{C}_R+\hat{\textbf{C}}_R\right)^{-1}\end{array}\right]\textbf{P}^T,
\end{eqnarray*}
is a  purely diagonal matrix, however if we treat it as a block-diagonal matrix, then its $i$-th  $2 \times 2$ diagonal block is 
\begin{eqnarray*}
\left[\begin{array}{cc}  \textbf{C}_{Dii}^{-1} & 0\\0 & \left(\textbf{C}_{Rii}+\hat{\textbf{C}}_{Rii}\right)^{-1}\end{array}\right] & = & \left[\begin{array}{cc}  \frac{1}{N_D(\omega_i)} & 0\\0 & \frac{1}{N_R(\omega_i)+\hat{N}_R(\omega_i)}\end{array}\right],
\end{eqnarray*}
whereas
\begin{eqnarray*}
\textbf{P}\left[\begin{array}{cc} \textbf{D}_{SD}\Psi_S\textbf{D}_{SD}^{\dag} & \textbf{D}_{SD}\Psi_S\textbf{D}_{SR}^{\dag}\\\textbf{D}_{SR}\Psi_S\textbf{D}_{SD}^{\dag} & \textbf{D}_{SR}\Psi_S\textbf{D}_{SR}^{\dag}\end{array}\right]\textbf{P}^T
\end{eqnarray*}
is a $2n \times 2n$ block diagonal matrix where the $i$-th diagonal block is a $2\times 2$ matrix given by,
\begin{eqnarray*}
\left[\begin{array}{cc} \textbf{D}_{SDii}\Psi_{Sii}\textbf{D}_{SDii}^* & \textbf{D}_{SDii}\Psi_{Sii}\textbf{D}_{SRii}^*\\\textbf{D}_{SRii}\Psi_{Sii}\textbf{D}_{SDii}^* & \textbf{D}_{SRii}\Psi_{Sii}\textbf{D}_{SRii}^*\end{array}\right] & = & \left[\begin{array}{cc} \left|H_{SD}(\omega_i)\right|^2P_S(\omega_i) & H_{SD}(\omega_i)H_{SR}^*(\omega_i)P_S(\omega_i)\\ H_{SD}^*(\omega_i)H_{SR}(\omega_i)P_S(\omega_i)& \left|H_{SR}(\omega_i)\right|^2P_S(\omega_i)\end{array}\right].
\end{eqnarray*}
Thus, the channel is decoupled, as the inputs corresponding to different coordinates do not interfere. Thus we have,
\begin{eqnarray*}
C_{nCF}^c
& = & \frac{1}{n}\sum_{i=1}^n I(x_{Si};y_{Di}^{c},\hat{y}_{Ri}^{c}|x_{Ri}),
\end{eqnarray*}
Similarly it can be shown for the compression rate constraint (\ref{eq:cf4}). With the decoupled channel, it is very easy to derive the expression in Theorem \ref{thm:cf}. 
\end{proof} 
\par As obvious from the Lemma \ref{lem:opt_CF} that the decomposition across the frequency band is also optimal for the CF achievable rate just as with the DF protocol. But note that in contrast to the case of DF, the $n$ parallel memoryless relay channels for CF are not only coupled via the input power constraints {\em but} also by the compression rate constraint.

\subsection{Capacity Upper Bound}
\par For the upper bound, we generalize the max-flow-min-cut theorem stated in \cite{Cover79, Cover_book} for the block memoryless relay channel,
\begin{eqnarray*}
C_n^{cup}(P_S,P_R) = \sup_q \frac{1}{n} \min\{C_{1n}^{cup},C_{2n}^{cup}\} \; \; \mbox{where}
 \; \;  \begin{array}{l}
C_{1n}^{cup}= I(x_S^n;{y}_R^{nc},{y}_D^{nc}|x_R^n),\\
C_{2n}^{cup}  =I(x_S^n,x_R^n;{y}_D^{nc}) \end{array}.
\end{eqnarray*}
The maximization is taken over the same input distribution as in the lower bound (see Lemma \ref{lemma:ninoslav}). For the Gaussian relay channel, using  similar matrix manipulations as done in the cases of the achievable rates, we can obtain,
\begin{eqnarray}
C_{1n}^{cup} & \leq & \sum_{i=1}^n \frac{1}{2} C\left(\alpha(\omega_i)P_S(\omega_i)\left(\frac{\left|H_{SR}(\omega_i)\right|^2}{N_R(\omega_i)}+\frac{\left|H_{SD}(\omega_i)\right|^2}{N_D(\omega_i)}\right)\right),
\end{eqnarray}
\begin{eqnarray}
C_{2n}^{cup} & = & \frac{1}{n} \sum_{i=1}^n \frac{1}{2}C \left(\frac{P(\omega_i)}{N_{D}(\omega_i)} \right),
\end{eqnarray}
where, $P(\omega_i)$ is as defined in the Theorem 1. This result implies that decomposition of the network into parallel scalar relay channels can also be effectively employed in calculating an upper bound on capacity.
\subsection{Optimality of the achievable rates}
\par We next examine scenarios under which CF and DF are capacity achieving.  To this end, we generalize the definition of degradedness \cite{Cover79} to the $n$-block memoryless relay channel as,
\begin{mydef}
A $n$ block memoryless vector relay channel $(\mathcal{X}_S^n \times \mathcal{X}_R^n, p(y_R^{nc},y_D^{nc}|x_S^n,x_R^n),\mathcal{Y}_R^{nc} \times \mathcal{Y}_D^{nc})$ is said to be degraded if,
\begin{eqnarray*}
p(y_R^{nc},y_D^{nc}|x_S^n,x_R^n) & = & p(y_R^{nc}|x_S^n,x_R^n)p(y_D^{nc}|x_R^n,y_R^n).
\end{eqnarray*}
\end{mydef}
An  alternative  statement of Definition 1 holds for the n-CGRC with ISI if we exploit properties of the DFT matrix and undertake some matrix manipulation:
\begin{mydef}
A $n$-block memoryless Gaussian circular relay channel is said to be \textit{degraded} if the following condition is satisfied,
\begin{eqnarray*}
\frac{\left|H_{SR}(\omega_i)\right|^2}{N_R(\omega_i)} & \geq & \frac{\left|H_{SD}(\omega_i)\right|^2}{N_D(\omega_i)},\;\; \forall \;\; 1\leq i \leq n.
\end{eqnarray*}
\end{mydef}
\noindent
where variables are defined as before. Thus, if the source-to-relay channel $\mbox{SNR}$ is better than the source-to-destination $\mbox{SNR}$ in all of the frequency sub-bands, then the circular relay channel is degraded. If the relay channel is degraded, then it can be readily seen that (see \cite{Cover79}) the achievable rate using DF coding coincides with the cut-set upper bound.
\par  In contrast, with the compress-and-forward coding scheme, capacity is achieved if $\hat{y}_R^{nc}$ is a deterministic, invertible function of the relay input $y_R^{nc}$.
Appendix B provides a proof of this condition. For example, this condition can be realized when the relay and destination are co-located.  In that case, the relay does not have to quantize its input $y_R^{nc}$ ($\hat{y}_R^{nc}=y_R^{nc}$) as the ultimate receiver at the destination can decode the uncompressed relay input because of physical proximity and hence the CF coding rate achieves the cut-set upper bound. This result is in accordance with the corresponding results on memoryless relay channels (see \cite{Kramer:2005}) and will be emphasized with examples later in the paper.
\section{Limiting Capacity of the Relay Channel with memory}
\par As implied by the results in \cite{Andrea01}, our desired capacity bounds  for the linear Gaussian relay channel with finite memory will be obtained if we take the limit as $n$ goes to infinity for the $n$-CGRC capacity bounds.  However, due to the presence of the relay, taking such limits requires careful treatment. 
In this section,  we take such limits to yield the primary results of the paper.
As before, we define the following power constraint set
\begin{eqnarray*}
{\cal P}^C & = & \left \{\alpha(\cdot),P_S(\cdot), P_R(\cdot): 0\leq\alpha(\omega)\leq 1, \frac{1}{2\pi}\int_{-\pi}^{\pi}P_S(\omega)d\omega\leq P_S, \frac{1}{2\pi}\int_{-\pi}^{\pi}P_R(\omega)d\omega\leq P_R \right\}.
\end{eqnarray*}

Thus, we have,
\begin{theorem}\label{thm:DFC}
The achievable rate for a LGRC with finite memory $m$, where the relay uses decode-and-forward coding strategy, is given by,
\begin{eqnarray*}
C_{DF}(P_S,P_R)& = & \max_{{\cal P}^C} \min \{C_{1DF}, C_{2DF}\} \; \; \; 
\mbox{where} \;\; \; \begin{array}{l}
C_{1DF} =\frac{1}{4\pi}\int_{-\pi}^{\pi}C \left(\frac{\alpha(\omega)\left|H_{SR}(\omega)\right|^2P_S(\omega)}{N_{R}(\omega)}\right) d\omega,\\
C_{2DF} = \frac{1}{4\pi}\int_{-\pi}^{\pi}C \left(\frac{P(\omega)}{N_{D}(\omega)} \right)d\omega, \end{array}
\end{eqnarray*}
and\\
\begin{eqnarray*}
P(\omega) & = & \left|H_{SD}(\omega)\right|^2P_S(\omega)+\left|H_{RD}(\omega)\right|^2P_R(\omega)+ 2\sqrt{\bar{\alpha}(\omega)\left|H_{SD}(\omega)H_{RD}(\omega)\right|^2P_S(\omega)P_R(\omega)}.
\end{eqnarray*}
\end{theorem}
\par The proof is found in Appendix C. The cut-set bound for the relay channel with memory is similarly derived, yielding the result below,
\begin{theorem}\label{thm:CutsetC}
An upper bound on the capacity of a LGRC with memory m is given by,
\begin{eqnarray*}
C^{up}(P_S,P_R)& = & \max_{{\cal P}^C} \min \{C_{1}^{up}, C_{2}^{up}\} \; \; \; 
\mbox{where} \; \; \; \begin{array}{l} C_{1}^{up}=\frac{1}{4\pi}\int_{-\pi}^{\pi}C \left(\alpha(\omega)P_S(\omega)\left(\frac{\left|H_{SR}(\omega)\right|^2}{N_R(\omega)}+\frac{\left|H_{SD}(\omega)\right|^2}{N_D(\omega)}\right)\right) d\omega,\\
C_{2}^{up} = C_{2DF} , \end{array}
\end{eqnarray*}
where $C_{2DF}$ is as given in Theorem 3.
\end{theorem}
Finally, we have the associated CF result:
\begin{theorem}\label{CFC}
The achievable rate for a LGRC with finite memory $m$, where the relay uses compress-and-forward coding strategy, is given by,
\begin{eqnarray*}
C_{CF}(P_S,P_R) & = & \sup  \frac{1}{4\pi}\int_{-\pi}^{\pi} C\left(P_S(\omega)\left(\frac{\left|H_{SD}(\omega)\right|^2}{N_D(\omega)}+\frac{\left|H_{SR}(\omega)\right|^2}{N_R(\omega)+\hat{N}_R(\omega)}\right)\right)d\omega,
\end{eqnarray*}
subject to the input power constraints as in the earlier bounds and,
\begin{eqnarray*}
\int_{-\pi}^{\pi} \log \hat{N}_{R}(\omega)d\omega & \geq & \int_{-\pi}^{\pi} \log\left( \frac{P_S(\omega)\left(\left|H_{SR}(\omega)\right|^2N_D(\omega)+\left|H_{SD}(\omega)\right|^2\hat{N}(\omega)\right)+\hat{N}(\omega)N_D(\omega)} {\left|H_{SD}(\omega)\right|^2P_S(\omega)+\left|H_{RD}(\omega)\right|^2P_R(\omega)+N_D(\omega)}\right)d\omega.
\end{eqnarray*}
\end{theorem}
\begin{proof}
It is easy to see that the input power constraints and the inequality (\ref{eq:Comp_rate_cons}) form a closed but non-convex constraint set. But as the objective function is continuous and bounded in $[-\pi,\pi]$ and concave in $P_S(\omega)$ and convex in $\hat{N}_R(\omega)$, by the properties of the Riemann integral Eqns. (\ref{eq:CFrate}) and (\ref{eq:Comp_rate_cons}) converge to the desired result.
\end{proof}
\section{Power Allocation}
In this section, we consider the input power spectral densities (PSDs) for our bounds on the capacity derived in the previous sections (Theorems 3-5). As the maximizing PSDs appear to be intractable to compute, for both the DF and CF protocol, approximations will be considered to yield closed form solutions. To simplify the notation, we let the $\mbox{SNR}$ of each link in each of the sub-bands be defined as follows,
\begin{eqnarray*}
a_{SRi}=\frac{|H_{SR}(\omega_i)|^2}{N_R(\omega_i)},a_{SDi}=\frac{|H_{SD}(\omega_i)|^2}{N_D(\omega_i)},a_{RDi}=\frac{|H_{RD}(\omega_i)|^2}{N_D(\omega_i)}.
\end{eqnarray*}

\subsection{Power allocation for DF}
\par We assume a total power constraint: $
\frac{1}{2\pi}\int_{-\pi}^{\pi}(P_S(\omega)+P_R(\omega))d\omega \leq P_t.$  The direct maximization over the power allocation appears to be intractable, thus we consider a lower bound on the achievable rate of Theorem \ref{thm:DFC} by exchanging the integration and infimum operations as,
\begin{eqnarray*}
C_{DF} & = & \max \min  \left\{\int C_{1}(\omega),\int C_{2}(\omega) \right\}d\omega\geq \max \int \min \left\{C_{1}(\omega), C_{2}(\omega)\right\}d\omega,
\end{eqnarray*}
where equality occurs if and only if $C_{1}(\omega)\leq C_{2}(\omega), \;\; \forall \omega$ and vice versa. Theorem \ref{thm:maximin} of Appendix C enables us to consider the power allocation problems in the continuous frequency domain via a finite dimension domain.  We then take the limit as $n\rightarrow\infty$ to get the desired result.  Exchanging the integration and infimum operations further simplifies the $n$-dimensional optimization problem into $n$-parallel optimization problems corresponding to each of the frequency sub-bands, which will be solved subsequently.  We observe that the resultant solution will be of a water-filling form.

\par It can be shown that the optimization problem at each of the sub-bands is a convex optimization problem, which will yield the following solution,
\begin{theorem}
If $a_{SRi}<a_{SDi}$
\begin{eqnarray*}
P_{t}^*(\omega_i) =\left(\nu_t-\frac{1}{a_{SRi}}\right)^+\\
P_{S}^*(\omega_i)= P_{t}^*(\omega_i),P_{R}^*(\omega_i)=0;
\end{eqnarray*}
And if $a_{SRi}\geq a_{SDi}$
\begin{eqnarray*}
P_{t}^*(\omega_i) & = & \left(1+\frac{a_{SRi}-a_{SDi}}{a_{SDi}+a_{RDi}}\right)\left(\nu_t-\frac{1}{a_{SRi}}\right)^+\\
P_{S}^{*}(\omega_i) & = & \left(1+a_{SDi}\frac{a_{SRi}-a_{SDi}}{(a_{SDi}+a_{RDi})^2}\right)\left(\nu_t-\frac{1}{a_{SRi}}\right)^+\\
P_{R}^*(\omega_i) & = & a_{RDi}\frac{a_{SRi}-a_{SDi}}{(a_{SDi}+a_{RDi})^2}\left(\nu_t-\frac{1}{a_{SRi}}\right)^+;\\
& & \mbox{Such that,}\;\;\;\; \frac{1}{n}\sum_{i=1}^n P_{t}^*(\omega_i) = P_t,
\end{eqnarray*}
where, $x^+=\max\{0,x\}$.
\end{theorem}
\par The proof uses the concavity of the objective function and standard convex optimization techniques and is summarized in Appendix D. For our modified cost function, the optimal input PSDs derived above all have a water-filling type of structure.  Unfortunately, no such
simple water-filling type of solution  appears to exist for the optimal power allocation problem associated with the upper bound. A more detailed and complex solution to this problem is discussed in \cite{Anders}.

\subsection{Power allocation for CF}
\par In the achievable rate using CF coding, to simplify the inequality constraint, we make the following substitutions,
\begin{eqnarray*}
\hat{N}_{Ri} = \frac{1+P_{Si}(a_{SRi}+a_{SDi})}{1+a_{SDi}P_{Si}}c_i, & &
q_i = \frac{\hat{N}_{Ri}}{1+c_i},
\end{eqnarray*}
where $c_i$ is a non-negative number. It can be easily shown that the CF achievable rate after the substitution is given by
\begin{eqnarray*}
C_{nCF}^c  & = & \sup \sum_{i=1}^n \frac{1}{2} \log \frac{(1+P_{Si}(a_{SRi}+a_{SDi}))^2}{1+P_{Si}(a_{SRi}(1+q_i)+a_{SDi})},
\end{eqnarray*}
where the supremum is taken over the input power constraints and the following inequality,
\begin{eqnarray*}
\sum_{i=1}^n \log q_i \geq \sum_{i=1}^n \log \frac{1+P_{Si}(a_{SRi}+a_{SDi})}{1+P_{Si}a_{SDi}+P_{Ri}a_{RDi}}.
\end{eqnarray*}
\par The new parameter $q_i$ is thus,
\begin{eqnarray}
q_i & = & \frac{\hat{N}_R(\omega_i)}{1+\frac{\hat{N}_R(\omega_i)(1+a_{SDi}P_S(\omega_i))}{1+P_S(\omega_i)(a_{SRi}+a_{SDi})}}\nonumber\\
& = & \frac{\hat{N}_R(\omega_i)(1+P_S(\omega_i)(a_{SRi}+a_{SDi}))}{(1+\hat{N}_R(\omega_i))+P_S(\omega_i)(a_{SRi}+(1+\hat{N}_R(\omega_i))a_{SDi})}.\label{eq:qi}
\end{eqnarray}
Now, as the inputs and additive noises are Gaussian random variables, Eqn (\ref{eq:qi}) can be re-written as,
\begin{eqnarray*}
\lefteqn{\frac{1}{2} \log (2\pi e)Var[q_i] = h(\hat{Y}_{Ri}^c|Y_{Ri}^c)+h(Y_{Di}^c,Y_{Ri}^c|X_{Ri})-h(\hat{Y}_{Ri}^c,Y_{Di}^c|X_{Ri})}\\
& \stackrel{(a)}{=} & h(\hat{Y}_{Ri}^c|Y_{Ri}^c,X_{Ri}, Y_{Di}^c)+h(Y_{Di}^c,Y_{Ri}^c|X_{Ri})-h(\hat{Y}_{Ri}^c,Y_{Di}^c|X_{Ri})\\
& = & h(\hat{Y}_{Ri}^c|Y_{Ri}^c,X_{Ri}, Y_{Di}^c)+h(Y_{Di}^c|X_{Ri})+h(Y_{Ri}^c|X_{Ri},Y_{Di}^c)-h(Y_{Di}^c|X_{Ri})-h(\hat{Y}_{Ri}^c|X_{Ri},Y_{Di}^c)\\
& = & h(Y_{Ri}^c, \hat{Y}_{Ri}^c|X_{Ri},Y_{Di}^c)-h(\hat{Y}_{Ri}^c|X_{Ri},Y_{Di}^c)\\
& = & h(Y_{Ri}|\hat{Y}_{Ri}^c,X_{Ri},Y_{Di}^c),
\end{eqnarray*}
where, (a) follows from the fact that $\hat{Y}_{Ri}^c$ is a quantized version of $Y_{Ri}^c$ and thus given $Y_{Ri}^c$, $\hat{Y}_{Ri}^c$ is independent of every other random variable. Thus, $q_i= Var[Y_{Ri}^c|\hat{Y}_{Ri}^c,X_{Ri},Y_{Di}^c]$ and can be interpreted as the variance of the backward test channel with the side information $(X_{Ri},Y_{Di}^c)$ or $Y_{Di}^{c'}=Y_{Di}^c - \sqrt{a_{RDi}}X_{Ri}$.
\par It is easy to see that the optimization problem described above is a non-convex one. The necessary conditions for the optimal solution using the Karush-Kuhn-Tucker (KKT) conditions are summarized in \cite{Chang} and the solution is given below for completeness,
\begin{eqnarray*}
\lefteqn{\frac{a_{SRi}(1-q_i^*)+(a_{SDi}+a_{SRi})(a_{SDi}+a_{SRi}(1+q_i^*))P_S^*(\omega_i)}{2(1+(a_{SDi}+a_{SRi})P_S^*(\omega_i))(1+(a_{SDi}+a_{SRi}(1+q_i^*))P_S^*(\omega_i))}}\\
& & - \frac{\lambda_1^*(a_{SRi}+(a_{SRi}+a_{SDi})a_{RDi}P_R^*(\omega_i))}{2(1+(a_{SDi}+a_{SRi})P_S^*(\omega_i))(1+a_{SDi}P_S^*(\omega_i)+a_{RDi}P_R^*(\omega_i))}-\lambda_2^* = 0,
\end{eqnarray*}
\begin{eqnarray}\label{eq:power_CF}
P_R^*(\omega_i) & = & \left[\frac{\lambda_1^*}{2\lambda_3^*}-\frac{1}{a_{RDi}}(1+a_{SDi}P_S^*(\omega_i))\right]^+,\nonumber\\
q_i^* & = & \left[\frac{\lambda_1^*}{1-\lambda_1^*}\left(1+\frac{1+a_{SDi}P_S^*(\omega_i)}{a_{SRi}P_S^*(\omega_i)}\right)\right]^+.
\end{eqnarray}
\par It is notable that $P_R^*(\omega_i)$ is given by a water-filling solution and independent of $a_{SRi}$, {\em i.e.}, the relay power allocation does not depend upon the characteristics of its incoming link. We can also see that $q_i^*$ is independent of $a_{RDi}$ and decreases as $a_{SRi}$ increases, which implies that the relay bit-rate allocation does not depend upon the characteristics of its outgoing link.

\par As is obvious from Eqn (\ref{eq:power_CF}), the optimal power allocation does not admit a classical water-filling form as in the case of the modified DF strategy discussed earlier. But by following \cite{Lee}, a much simpler and a convex optimization problem can be obtained by choosing,
\begin{eqnarray*}
q_i & = & \frac{1+(a_{SDi}+a_{SRi})P_S(\omega_i)}{1+a_{SDi}P_S(\omega_i)+a_{RDi}P_R(\omega_i)}.
\end{eqnarray*}
It can be easily shown that under this condition the achievable rate will be,
\begin{eqnarray*}
C_{CFm} & = & \sup_{P_S(\omega_i),P_R(\omega_i)} \sum_{i=1}^n \frac{1}{2}C\left(a_{SDi}P_S(\omega_i)+\frac{a_{SRi}P_S(\omega_i)}{1+\frac{1+(a_{SDi}+a_{SRi})P_S(\omega_i)}{a_{RDi}P_R(\omega_i)}}\right).
\end{eqnarray*}
such that the power constraints are satisfied. However,  $C_{nCF}^c\geq C_{CFm}$.

\section{Illustrative Examples}
\par In this section, we discuss several simple examples to illustrate the computation and relationships of different achievable rates and the upper bound. We will start with simple example of relay channel with equal transmission bandwidth on all the links to illustrate how the theorems are applied and then move on to other, more practical, examples.

\begin{figure}
	\centering
		\includegraphics[width=0.50\textwidth]{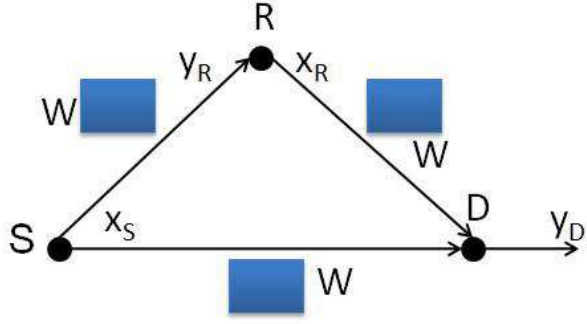}
	\caption{Relay with equal bandwidth}
	\label{fig:same_bandwidth}
\end{figure}

\subsection{Relay with equal bandwidths} 
\par We examine the simplest case when all the channels are the same ideal low-pass filters of bandwidth $W$ and the noise is AWGN (see Fig. 3). Let,
\begin{eqnarray*}
N_R(\omega)= N_1, N_D(\omega)=N_1+N_2=N, \hspace{0.1cm}0\leq\omega\leq\frac{W}{2\pi}.   
\end{eqnarray*}
Due to the common channel assumption and the noise variances, the relay channel is degraded. We make the following idealized assumption: we can achieve both strict bandwidth and time limitation simultaneously, the key is that we require finite memory. It can be shown that the channel capacity is given by,
\begin{eqnarray*}
C =\max_{0\leq \alpha \leq 1}\min W \left\{C \left(\frac{\alpha P_S}{N_1W} \right),C\left(\frac{P_S+P_R+2\sqrt{\bar{\alpha}P_SP_R}}{NW}\right) \right \}.
\end{eqnarray*}
\par Uniform input PSD's achieve capacity, which is expected since all the parallel degraded relay channels are identical. This result is a generalization of the discrete memoryless Gaussian relay channel \cite{Cover79} in the bandwidth limited case. If $\frac{P_S}{N_1}\leq \frac{P_R}{N_2}$, then $\alpha=1$ maximizes the capacity and its given by $W C \left(\frac{P_S}{N_1W} \right)$ and if $\frac{P_S}{N_1}> \frac{P_R}{N_2}$ then the capacity is given by $W C \left(\frac{\alpha^* P_S}{N_1W} \right)$, where $\alpha^*$ is solution of $\frac{\alpha P_S}{N_1W}=\frac{P_S+P_R+2\sqrt{\bar{\alpha}P_SP_R}}{NW}$.

\begin{figure}
	\centering
		\includegraphics[width=0.80\textwidth]{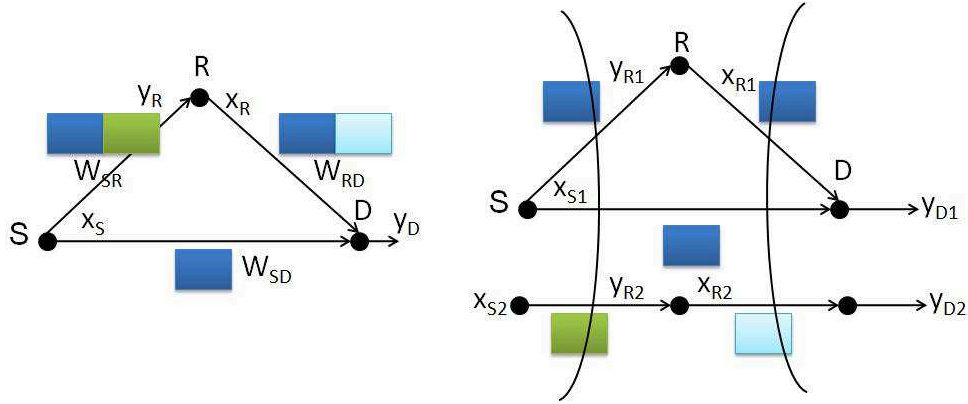}
	\caption{Decomposition of a relay with unequal bandwidths into a relay with equal bandwidth and a two hop channel with remaining bandwidths}
	\label{fig:two_hop}
\end{figure}

\subsection{Relay with unequal bandwidths}

\par Our second example considers a single-relay channel with different bandwidths across different links. This type of channel is common in underwater communication where channel bandwidth depends on internode separation (see \cite{Wenyi08, Wenyi10}). Let, $W_{SD},W_{SR}$ and $W_{RD}(W_{SD}<W_{SR},W_{SD}<W_{RD})$ be the bandwidths of the three links (see Fig. 4). Furthermore, we assume that all channels are ideal lowpass filters. Noise is AWGN and has the same PSD as in the earlier example. It can be readily seen that the relay channel is degraded and it can be shown that the network reduces to a degraded relay of bandwidth $W_{SD}$ and a 2-hop channel with link bandwidths $W_{SR}-W_{SD}$ and $W_{RD}-W_{SD}$. Hence the capacity of this channel is given by,
\begin{eqnarray*}
C & = & \min W_{SD} \left \{C \left(\frac{\alpha P_{S1}}{N_1W_{SD}} \right),C\left(\frac{P_{S1}+P_{R1}+2\sqrt{\bar{\alpha}P_{S1}P_{R1}}}{NW_{SD}}\right)\right \}\\
& & +\min \left\{(W_{SR}-W_{SD})C \left(\frac{P_{S2}}{N_1(W_{SR}-W_{SD})} \right),
 (W_{RD}-W_{SD})C \left(\frac{P_{R2}}{N(W_{RD}-W_{SD})}\right)\right \},
\end{eqnarray*}
where, $P_{S1}+P_{S2}\leq P_S, P_{R1}+P_{R2}\leq P_R$. The decomposition is very similar to the decomposition observed in \cite{ashu_it} for the rate-constrained relay channel, where because of the constraint on the relay encoding rate, the source splits its rate between direct transmission and cooperative transmission using relay. In our example, due to the excess bandwidth available on the $2$ hop link, the source splits the rate between the two parallel sub-channels.

\subsection{Suboptimality of random permutation at relay}

\par In this simple example, we show that when sending indepepndent messages (the covariance matrices of $X_S^n$ and $X_R^n$ are diagonal) across the sub-carriers, the best permutation at the relay for the channel decomposition is to match the $i$-th sub-carrier of the source-relay channel to the $i$-th subcarrier of the relay-destination channel, i.e., any other permutation at the relay would be as good as no permutation at all. Consider a parallel relay channel with $2$ sub-carriers, i.e.,
\begin{eqnarray}
\textbf{H}_q & = & \left[\begin{array}{cc} H_{q1} & 0\\0 & H_{q2}\end{array}\right],
\end{eqnarray}
where, $q\in\{SR, SD, RD\}$. Let the relay employ DF processing, .i.e., the relay decodes the received messages in each of the sub-carriers of the source-relay channel and then according to a deternministic permutation function $\pi:\{1,2,\cdots,n\}\mapsto \{1,2,\cdots,n\}$ ($n=2$ in this case), the relay sends the decoded message of the $i$-th sub-carrier of source-relay channel in the $j$-th sub-carrier of the relay-destination channel, .i.e., $\pi(i)=j$.

\begin{figure}
	\centering
		\includegraphics[width=0.60\textwidth]{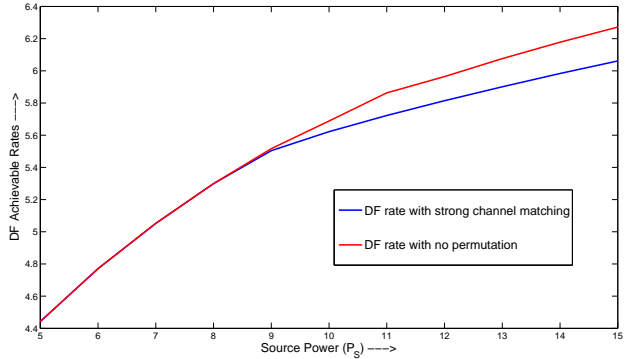}
	\caption{DF achievable rates with a) no permutation at the relay and b) strong channel matching at the relay as source power $P_S$ varies from $5$ to $15$. Relay and destination noises are Gaussian and of unit variance. $|H_{SR}|^2=[1.5 1.8], |H_{SD}|^2=[0.7 0.5],|H_{RD}|^2=[0.8 1]$ and $P_R=10$.}
	\label{fig:two_hop}
\end{figure}

\par In our example, two different permutation functions are possible, a) $\pi(i)=3-i$, and b) $\pi(i)=i$, for $i=1,2$. We will compare the DF achievable rate of both the cases for an example channel and will see that it is enough to consider the permuting function $\pi(i)=i$ for the relay (see Fig. 5).

\par We observe in Fig. 5 that the two achievable rates match for low values of source power $P_S$. This is due to the fact that if the relay power $P_R$, exceeds that of the source, the relay to destination link can fully support the rate of the source to relay channel. Owing to the weak source to destination link - $|H_{SD}|^2$ is smaller compare to other two channel transfer functions - the overall relay channel effectively behaves like a two-hop channel without the source to destination link. This verifies the fact that the channel matching with a modified amplify-and-forward works well in multihop channels (see \cite{wenyi}).    

\par We explain the loss in performance due to relay permutation as follows.  If we permute the relay inputs, the only affected term in the DF achievable rate expression in Theorem \ref{thm:df} is $C_{2nDF}^c$, which increases when the source and relay cooperatively send the bin index of the message to the destination. By sending the $i$-th sub-carrier relay input through  the $j$-th sub-channel in the relay-destination link ($i\neq j$), one cannot exploit the potential cooperative gain in a single sub-channel.

\subsection{Comparison of achievable rates}

\par The example in this subsection is motivated by underwater acoustic communication. UW channels are characterized by unique physical and statistical properties. The physical property that we are interested in is the attenuation which depends on the propagation distance and carrier frequency of the transmitted signal. For the statistical part, the channel is assumed to be WSSUS (Wide Sense Stationary with Uncorrelated Scattering). We model a UW channel by taking into account both properties to form a frequency-dependent fading multipath channel. 
\par \textit{Frequency-dependent Path Loss}: For signals propagated through UW medium, its energy is attenuated as a
function of both the distance $d$ and signal frequency $f$. This attenuation or path loss is a combination of the geometric
spreading and absorption and can be written as in \cite{milica2006} 
\begin{eqnarray*}
A(d,f) & = & A_0d^ka(f)^d.
\end{eqnarray*}
The spreading factor $k$ is set to 1.5. The absorption coefficient $a(f)$ can be expressed emperically in dB/Km for $f$ in KHz, using the Thorp's formula as \cite{Berkhovskikh}
\begin{eqnarray*}
10\log a(f) & = & 0.11 \frac{f^2}{1+f^2}+44\frac{f^2}{4100+f^2}+2.75\times10^{-4}f^2+0.003. 
\end{eqnarray*}

\par \textit{Statistical Channel Model:} We consider that a transmission bandwidth of $W>>(\Delta f)_c$ (coherence bandwidth) is available to the user, thus the channel is frequency selective. The channel is modeled as a tapped delay line, where the number of taps are equal to the number of multipaths present, and it is given by $L=\left\lceil \frac{(\Delta T)_c}{T}\right\rceil$, where $(\Delta T)_c =\frac{1}{(\Delta f)_c}$ is the delay spread of the channel calculated from the ray tracing model and $T=\frac{1}{W}$ is the rate of sampling of the input symbols. Let $h(t,\tau)$ denote a continuous-time channel impulse response (CIR) of linear time-variant (LTV) UW channel \cite{Proakis},
\begin{eqnarray*}
h(t,\tau) & = & \sum_{i=0}^{L-1}h_i(t)\delta(\tau-\frac{i}{W}).
\end{eqnarray*}
WSSUS is commonly assumed to characterize the UW channel \cite{Proakis, Kilfoyle} \textit{i.e.}, $E[h(t,\tau)h(t',\tau')]= R_h(t-t',\tau)\delta(\tau-\tau')$, where $R_h(t-t',\tau)$ is the autocorrelation function of the delay $\tau$ between time $t$ and $t'$. This implies that the $\{h_i(t)\}$ are mutually uncorrelated. We consider $\{h_i(t)\}$ to be Gaussian random processes, hence uncorrelated scattering assumption make them statistically independent. 

\par The ambient noise in the ocean can be modeled using four sources: turbulence, shipping, waves, and thermal noise. The following empirical formulae give the PSD of the four noise components in dB re $\mu$ Pa per Hz as a function of frequency
in kHz \cite{Coates}:
\begin{eqnarray}\label{eq:noise_var}
10\log N_t(f) & = & 17-30\log f\nonumber\\
10\log N_s(f) & = & 40+20(s-0.5)+26\log f-60\log(f+0.03)\nonumber\\
10\log N_w(f) & = & 50 +7.5w^{\frac{1}{2}}+20\log f-40\log(f+0.4)\nonumber\\
10\log N_{th}(f) & = & -15+20\log f.
\end{eqnarray}
\par Turbulence noise influences only the very low frequency region, $f < 10$ Hz. Noise caused by distant shipping is dominant in the frequency region 10 Hz -100 Hz, and it is modeled through the shipping activity factor $s$, whose value ranges between 0 and 1 for low and high activity, respectively. Surface motion, caused by wind-driven waves ($w$ is the wind speed in
m/s) is the major factor contributing to the noise in the frequency region 100 Hz - 100 kHz (which is the operating region used by the majority of acoustic systems). Finally, thermal noise becomes dominant for $f >100$ kHz.

\begin{figure}
  \centering
    \includegraphics[width=0.50\textwidth]{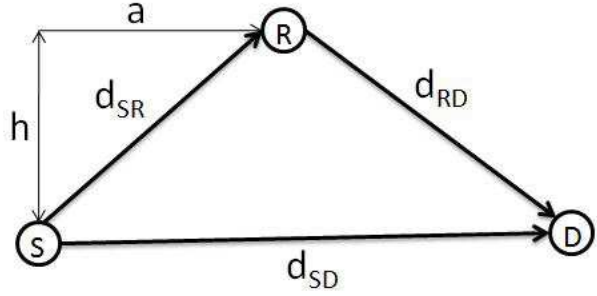}
  \caption{Channel model of a single-relay channel with ISI}\label{fig:relay}
\end{figure}

\par In this example, we compare DF achievable rate studied with respect to the direct transmission from source to destination and $2$ hop relaying with the total input power constraint remaining same in all the cases. For simulation, we have taken $(\Delta f)_c=3.33\; \mbox{KHz}$ for the SR and RD link and $(\Delta f)_c=5 \; \mbox{KHz}$ for the SD link. A carrier frequency of $f_c$ of $27\; \mbox{KHz}$ and an available transmission bandwidth of $10\; \mbox{KHz}$ are considered for all the 3 links. The physical attenuation of the channel is calculated using Thorp's formula. Rayleigh fading model is investigated where each channel tap is a complex Gaussian random process, whose variances sum up to $1$. Noise is colored Gaussian as defined in (\ref{eq:noise_var}), with $s=0$ and $w=10$m/s. A total input power constraint of $P_t=20\; \mbox{dB}$ is considered for all the schemes. We also assume that the the fading realizations are independent of each other and the channel state information is available at all the nodes. For the performance analysis, we average the achieved rate for a particular realization of the channel over all realizations of the fading states in order to capture the effects of fading. We use the channel model of Fig. 6.

\begin{figure}
  \centering
    \includegraphics[width=0.60\textwidth]{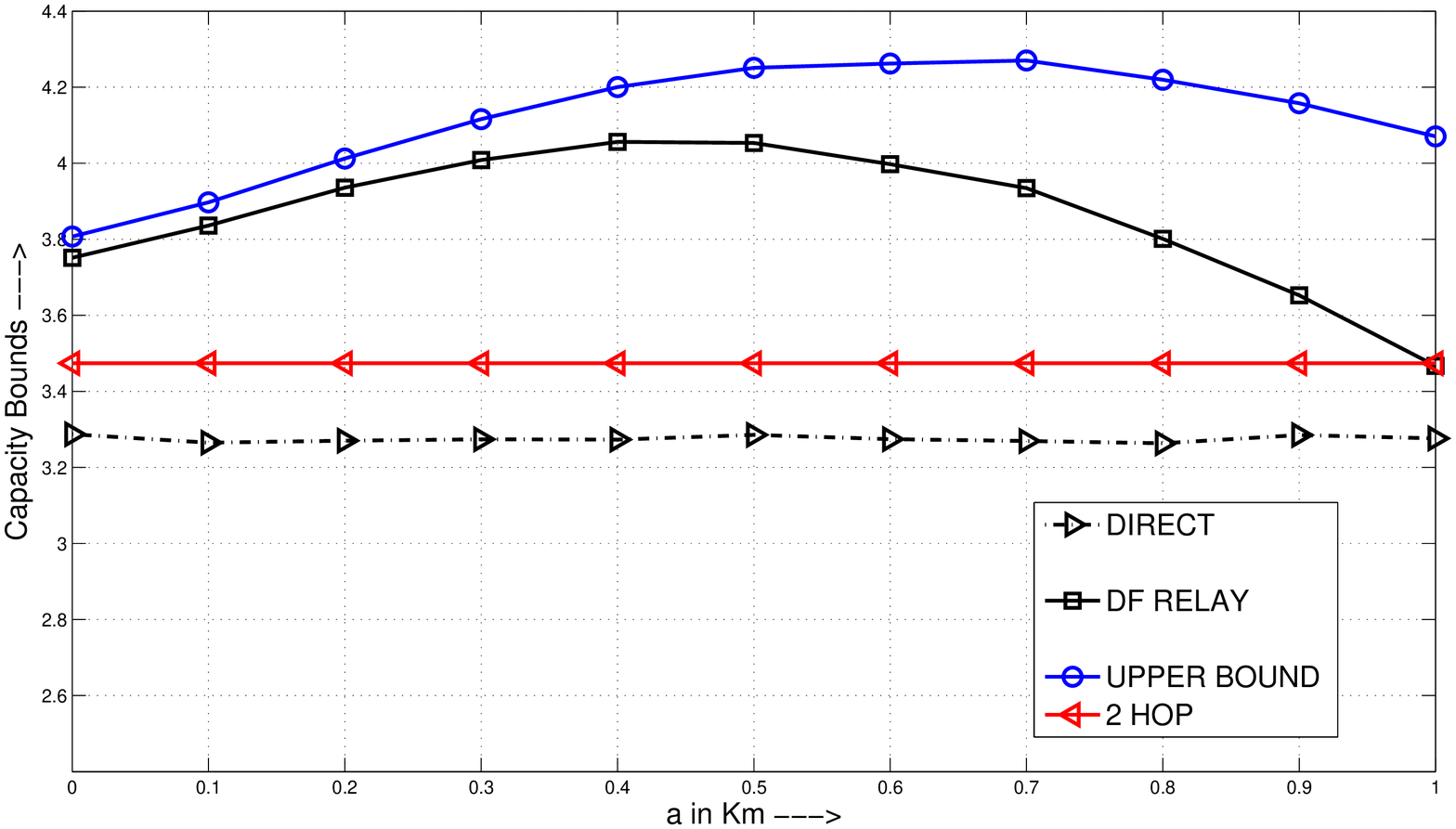}
  \caption{Achievable rates and upper bound when $h=0.25\;\mbox{Km},d_{SD}=1\;\mbox{Km}$ and $a$ varies from $0$ to $1.0\;\mbox{Km}$}\label{fig:a_vary}
\end{figure}

\par Fig. 7 plots the capacity bounds in bits/sec/Hz as we vary the relay position $a$. Comparing $2$ hop, direct transmission with the cooperative DF scheme, it can be concluded that relay cooperation performs best in terms of achievable rate in underwater acoustic communication as expected but given a fixed square area/volume and a set of $n>3$ nodes, the placement of the nodes that gives maximum rate between a given pair of source and destination, is still an open question and is worth investigating.

\subsection{Symbol asynchronous relay channel}
\par A relay channel is said to have asynchronism among the nodes if the codewords transmitted from the aource and relay do not coincide in time at the reciever. Frame synchronism refers to the ability of the nodes to receive or initiate the transmission of their codewords in unison. In many practical situations, it is perfectly reasonable to assume that this type of synchronism is achievable with the help of channel feedback or cooperation among transmitters. In contrast, \textit{symbol synchronism} is far more challenging to achieve due to the smaller time scales. Our interest in asynchronous relay channels is motivated by underwater acoustic communication networks, the slow speed of propagation of sound in water implies that significant asynchronism can occur in multi-terminal underwater networks.

\par Asynchronism in multiuser networks has been examined in multiple contexts over the years from an information theoretic perspective; (see, {\em e.g.}, \cite{Hui, verdu} and references therein). In this paper, we find the bounds on the capacity of a symbol-asynchronous Gaussian single-relay channel in which the node $i$ linearly modulates its symbols employing a fixed waveform $\xi_i(t)$-- this could be a signature code or a pulse shape. We exploit the results of \cite{verdu} , which examined the symbol asynchronous Gaussian MAC, to show that our symbol asynchronous relay network is equivalent to a relay with finite memory. With this equivalence in hand, we can exploit our prior results for relay channels with inter-symbol interference in Section IV to determine closed form expressions for upper and lower bounds on the capacity of the symbol asynchronous relay channel.

\subsubsection{Channel Model}
\par For a single-relay channel, assuming frame synchronism and additive white Gaussian noise model, we can write the continuous time received signals as,
\begin{eqnarray}\label{eq:asynch_cont}
y_{R}(t) & = & \sum_{i=1}^n x_{S}(i)\xi_S(t-iT-\tau_{SR})+n_{R}(t)\\
y_{D}(t) & = & \sum_{i=1}^n x_{S}(i)\xi_S(t-iT-\tau_{SD})+x_{R}(i)\xi_R(t-iT-\tau_{RD})+n_{D}(t),
\end{eqnarray}
where $\{x_{S}(i)\}_{i=1}^n, \{x_{R}(i)\}_{i=1}^n$ are the input sequences of the source and relay, respectively. $\xi_S(t), \xi_R(t)$ are the unit energy modulating waveforms of support $[0, T]$ used by the source and relay. The delays or offsets $\tau_{SR}, \tau_{SD}, \tau_{RD}$ account for the symbol asynchronism between the users and are known to the receiver. $n_R(t), n_D(t)$ are additive white Gaussian noise at the relay and destination with power spectral density equal to $\sigma_R^2$ and $\sigma_D^2$, respectively. We assume the same input power constraints as in ($\ref{eq:power_cons}$).

\par We can obtain an equivalent channel model with discrete-time outputs whose capacity is same as that of the continuous time one described above by considering the projection of the observation process $\{y_R(t), y_D(t)\}$ along the direction of the unit energy signals $\{\xi_S(t)\}$ and $\{\xi_R(t)\}$ and their T-shifts:     
\begin{eqnarray}\label{eq:proj}
y_R(i)& = & \int_{iT+\tau_{SR}}^{(i+1)T+\tau_{SR}}y_R(t)\xi_S(t-iT-\tau_{SR})dt \nonumber\\
y_{Dj}(i) & = &\int_{iT+\tau_{jD}}^{(i+1)T+\tau_{RD}}y_D(t)\xi_j(t-iT-\tau_{jD})dt, \nonumber\\
\end{eqnarray}
where, $j \in \{R,S\}$. By substituting ($\ref{eq:asynch_cont}$) into ($\ref{eq:proj}$) and by defining the cross-correlations between the assigned signature waveforms $\xi_S(t)$ and $\xi_R(t)$ as (assuming without loss of generality that $\tau_{RD}\leq \tau_{SD}$)
\begin{eqnarray*}
\rho_{RS} =\int_0^T \xi_R(t)\xi_S(t+\tau_{RD}-\tau_{SD})dt,\\
\rho_{SR} =\int_0^T \xi_R(t)\xi_S(t+T+\tau_{RD}-\tau_{SD})dt,
\end{eqnarray*}
it follows easily that the discrete-time channel output is given by,
\begin{eqnarray}\label{eq:asynch_discrete}
y_R(i)  & = & x_S(i)+n_R(i)\\
\left[\begin{array}{c} y_{DR}(i)\\y_{DS}(i)\end{array}\right] & = & \sum_{|j|\leq 1}\textbf{H}(j)\left[\begin{array}{c} x_R(i+j)\\x_S(i+j)\end{array}\right] + \left[\begin{array}{c} n_{DR}(i)\\n_{DS}(i)\end{array}\right],
\end{eqnarray}
where,$\textbf{H}(0), \textbf{H}(-1)=\textbf{H}^T(1)$ are given by,
\begin{eqnarray*}
\textbf{H}(0)= \left[\begin{array}{cc} 1 & \rho_{RS}\\\rho_{RS} & 1\end{array}\right], \textbf{H}(1)=\left[\begin{array}{cc} 0 & \rho_{SR}\\0 & 0\end{array}\right]
\end{eqnarray*}
and $1\leq i \leq n$ ($x_S(0)=x_S(n+1)=0, x_R(0)=x_R(n+1)=0$); the discrete-time noise process $\{[\begin{array}{cc} n_{DR}(i) & n_{DS}(i)\end{array}]^T\}$ is Gaussian with zero mean and covariance matrix:
\begin{eqnarray*}
E\left[\left[\begin{array}{c} n_{DR}(i)\\n_{DS}(i)\end{array}\right]\left[\begin{array}{cc} n_{DR}(j) & n_{DS}(j)\end{array}\right]\right] = \sigma_D^2 \textbf{H}(i-j).
\end{eqnarray*}

\begin{figure}
	\centering
		\includegraphics[width=0.40\textwidth]{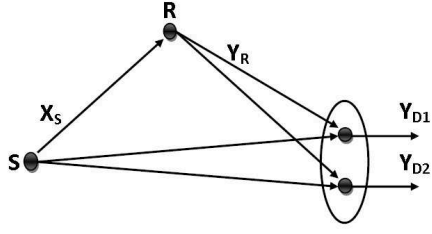}
	\caption{Discrete time equivalent of a symbol-asynchronous relay channel}
	\label{fig:relay_discrete}
\end{figure}

\par Since the receivers at the relay and destination know the assigned waveforms $\xi_S(t)$ and $\xi_R(t)$ as well as the symbol epochs $\{iT + \tau_q\}, \;\;q\in\{SR,SD,RD\}$, these receivers can compute $\{y_R(i)\}_{i=1}^n, \{y_{DR}(i)\}_{i=1}^n$ and $\{y_{DS}(i)\}_{i=1}^n$ by passing the observations through two matched filters for signals $\xi_S(t)$ and $\xi_R(t)$, respectively. The key observation is that this operation yields sufficient statistics for the transmitted messages \cite{Lehmann} and the equivalent channel model is a MIMO relay channel with memory $2$ (see Fig. 8). It is clear from the equivalent discrete time channel model that the capacity is independent of the delay of the signal coming from source to relay, as the channel impulse responses of the three links in the relay channel are functions of the relative offsets between the users in the MAC (multiple-access channel) portion of the relay. If either $\rho_{RS} = 1$ or $\rho_{SR} = 0$, then the channel becomes memoryless, as signals are symbol synchronous. For example, if the users are assigned the same signal and the channel is symbol synchronous, both outputs  coincide and are equal to
\begin{eqnarray*}
y_R(i) & = & x_S(i)+n_R(i)\\
y_D(i) & = & x_S(i)+x_R(i)+n_D(i).
\end{eqnarray*}
The channel is then a conventional discrete-time Gaussian relay channel, whose capacity is discussed in \cite{Cover79}. If the assigned signals are not equal, but the users remain symbol synchronous, then the outputs reduce to the memoryless MIMO relay channel,
\begin{eqnarray*}
y_R(i) & = & x_S(i)+n_R(i)\\
\left[\begin{array}{c} y_{DR}(i)\\y_{DS}(i)\end{array}\right] & = & \left[\begin{array}{cc} 1 & \rho\\\rho & 1\end{array}\right]\left[\begin{array}{c} x_R(i)\\x_S(i)\end{array}\right] + \left[\begin{array}{c} n_{DR}(i)\\n_{DS}(i)\end{array}\right],
\end{eqnarray*}
where $\{[n_{DR}(i) n_{DS}(i)]^T\}$ is an independent Gaussian process with covariance matrix given by,
\begin{eqnarray*}
\left[\begin{array}{cc} \sigma_D^2 & \sigma_D^2\rho\\\sigma_D^2\rho & \sigma_D^2\end{array}\right],
\end{eqnarray*}
and $\rho=\int_0^T \xi_S(t)\xi_R(t)dt$. The capacity of memoryless MIMO relay channel is studied in \cite{Anders05}.

\subsubsection{Achievable rates and upper bound}
\par We examine two relay coding strategies: a) Decode-and-Forward (DF) and b) Compress-and-forward (CF) and the "cut-set" upper bound. As CF and DF have differing regimes in which they offer the best rate \cite{Kramer:2005} for memoryless channels, it is of interest to investigate both coding strategies for the symbol-synchronous relay channel as we do herein.
\begin{theorem}\label{thm:DF_asynch}
The DF achievable rate of a symbol asynchronous relay channel with input power constraints ($\ref{eq:power_cons}$) is given by,
\begin{eqnarray*}
C_{DF}(P_S,P_R) & = & \sup \frac{1}{4\pi} \min \left\{\int_{0}^{2\pi}C\left( \frac{1}{\sigma_R^2}\alpha(\omega)P_S(\omega)\right)d\omega,\right.\\
& & \left. \int_{0}^{2\pi} C\left( \frac{1}{\sigma_D^2}(P_S(\omega)+P_R(\omega) +2\sqrt{(1-\alpha(\omega))P_S(\omega)P_R(\omega)}\rho(\omega)\right.\right.\\
& & \left.\left. +\frac{1}{\sigma_D^4}\alpha(\omega)P_S(\omega)P_R(\omega)(1-\rho^2(\omega))\right)d\omega\right\},
\end{eqnarray*}
where, $C(x)=log(1+x), \rho(\omega)=\rho_{RS}+\rho_{SR}cos(\omega)$ and $P_S(\omega)$ and $P_R(\omega)$ are the power allocated by the source and relay in the band $\omega$ and $\alpha(\omega)$ is the correlation between the source and relay codewords in the band $\omega$.
\end{theorem}
\par The theorem is proved in Appendix E. 
\begin{theorem}\label{thm:CF_asynch}
The CF achievable rate of a symbol asynchronous relay channel with input power constraints ($\ref{eq:power_cons}$) is given by,
\begin{eqnarray*}
C_{CF}(P_S,P_R) & = & \sup \frac{1}{4\pi} \int_{0}^{2\pi} \log \frac{A(\omega)}{1+\hat{N}_R(\omega)}d\omega\\
\int_{0}^{2\pi}\log \frac{A(\omega)}{B(\omega)}d\omega & \leq & 0,
\end{eqnarray*}
where, $A(\omega)=1+\hat{N}_{R}(\omega)+2P_{S}(\omega)+\hat{N}_{R}(\omega)P_{S}(\omega)$, $B(\omega)= \hat{N}_{R}(\omega)((1+P_{S}(\omega))(1+P_{R}(\omega))-P_{S}(\omega)P_{R}(\omega)\rho^2(\omega))$ and $\hat{N}_{R}(\omega)$ is the compression noise at the relay, which limits the amount of compression that can be performed at the relay.
\end{theorem}

\begin{theorem}\label{thm:Upper_asynch}
The upper bound on the capacity of a symbol asynchronous relay channel with input power constraints ($\ref{eq:power_cons}$) is given by,
\begin{eqnarray*}
C_{upper}(P_S,P_R) & = & \sup \frac{1}{4\pi} \min \left\{\int_{0}^{2\pi}C\left( \left(\frac{1}{\sigma_R^2}+\frac{1}{\sigma_D^2}\right)\alpha(\omega)P_S(\omega)\right)d\omega,\right.\\
& & \left. \int_{0}^{2\pi} C\left( \frac{1}{\sigma_D^2}(P_S(\omega)+P_R(\omega) +2\sqrt{(1-\alpha(\omega))P_S(\omega)P_R(\omega)}\rho(\omega)\right.\right.\\
& & \left.\left. +\frac{1}{\sigma_D^4}\alpha(\omega)P_S(\omega)P_R(\omega)(1-\rho^2(\omega))\right)d\omega\right\},
\end{eqnarray*}
\end{theorem}
\par The proof of these two Theorem $\ref{thm:CF_asynch}$ and $\ref{thm:Upper_asynch}$ are very similar to the Theorem $\ref{thm:DF_asynch}$ and thus omitted for brevity.  
\par The results can be easily extended to the case where the transmitters only know that the crosscorrelations $( \rho_{RS}, \rho_{SR})$ that parametrize the channel belong to an uncertainty set $\Gamma$, which is determined by the choice of the signature waveforms. For example, if both users are assigned a rectangular waveform then the uncertainty set is equal to the segment $\Gamma = \{0\leq \rho_{RS} \leq 1, 0\leq \rho_{SR} \leq 1, \rho_{SR}+\rho_{RS} = 1\}$. The achievable rate of the Gaussian asynchronous relay channel under this condition is obtained by taking infimum of the rates in Theorem $\ref{thm:DF_asynch}$ and $\ref{thm:CF_asynch}$ over the set $\Gamma$. This follows simply because for reliable communication a code has to be good no matter which actual channel is in effect.
                 
\par We will now compare DF and CF achievable rate for symbol-asynchronous relay and show the respective optimal regions with an example. We place the nodes such that the source, relay and destination are aligned, and the distances of the direct paths in each of the links are given by $d_{SR}=d, d_{RD}=1-d$ and $d_{SD}=1$. This corresponds to the relay channel of Fig. 6 with $h=0$ and $a=d$. The channel impulse response at each of the three links follows inverse power law, .i.e., $H_{SR}=d^{-\frac{\alpha}{2}}, H_{SR}=(1-d)^{-\frac{\alpha}{2}}$ and $H_{SD}=1$, where $\alpha$ is the attenuation constant. Modulating waveforms $\xi_S(t)$ and $\xi_R(t)$ are rectangular waveforms of unit energy and finite support $T$. For simplicity, the additive noises at the relay and destination are white, Gaussian and of unit power. The relative delays between the users are proportional to the distance, for e.g., $\tau_{SR}=Td, \tau_{RD}=T(1-d)$ and $\tau_{SD}=T$, which means if the relay is at the source, $\tau_{SR}=0$ and $\tau_{RD}=\tau_{SD}=T$, and similarly, if the relay and destination are co-located, $\tau_{SR}=\tau_{SD}=T$ and $\tau_{RD}=0$.       

\begin{figure}
	\centering
		\includegraphics[width=0.60\textwidth]{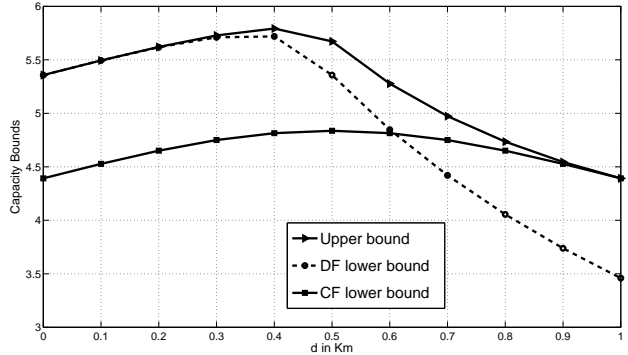}
	\caption{Rates for symbol-asynchronous relay network for $P_S=10, P_R=10, \alpha=2$ and varying $d$}
	\label{fig:Regime}
\end{figure}
 
\par With equal power constraints at the source and relay, we plot DF and CF achievable rates and upper bound on the capacity of this channel model in Fig. 9, as we move the relay from source to the destination. It can be seen that when the relay and destination are co-located ($d\rightarrow 1$), then relay can send its input directly to the destination without quantizing it and thus $\hat{N}_{Rii}=0, \;\; \forall i$. In this case CF rate performs optimally, matching the cut-set upper bound as $\hat{y}_R^{nc}=y_R^{nc}$, a deterministic function of $y_R^{nc}$. In contrast, when $d\rightarrow 0$, DF is optimal. Note that, when $d\rightarrow 0$ the channel is symbol-synchronous and memoryless as $\tau_{RD}=\tau_{SD}$ and $\xi_S(t)=\xi_R(t)$. These trends for CF and DF for the LGRC mimic those for the memoryless relay \cite{Kramer:2005}.

\section{Conclusions}
\par In this paper, we have derived single-letter expressions for the achievable rates and an upper bound on the capacity of a relay channel with inter-symbol interference and additive colored Gaussian noise. Such systems find wide application in a variety of wireless communication systems, in particular, underwater acoustic communication channels.  We have examined two important relay channel coding strategies, decode-and-forward and compress-and-forward; further, the upper bound is a generalization of the cut-set bound for multi-terminal networks.  Some of the conditions for which the channel capacity can be computed, such as degraded relay channels, are delineated.  The proof methods rely on the decomposition of the multipath channel into parallel channels via a DFT decomposition.  Thus, our results suggest the optimality of OFDM input signaling even for relay channels.  As such, optimal achievable schemes for memoryless channels are likely to have similar properties when extended to relay channels with finite memory.  Numerical examples were provided to illustrate properties of the results.  Ongoing work is extending these results to consider highly asynchronous links and developing practical coding strategies to achieve theoretical performance. Future work would be to investigate the Amplify-and-Forward (AF) strategy at the relay. Although, we numerically observe that decomposition into parallel channels is sub-optimal for AF relay, we can still probably come up with a simple linear filter at the relay, which will outperform the DF and CF achievable rates for some values of channel parameters.    

\appendices

\section{Proof of Lemma \ref{lem:opt_CF}}
\begin{proof}
To prove this Lemma, we need to show that Eqn. (\ref{eq:cf1}) subject to the power constraints (\ref{eq:power_cons}) and rate constraint (\ref{eq:cf4}) is maximized by a diagonal $\Psi_S$ and $\Psi_R$. We prove this by showing that diagonal $\Psi_S$ and $\Psi_R$ maximizes the objective function Eqn. (\ref{eq:cf1}) without altering the constraint set (\ref{eq:cf4}). Consider the first term in Eqn. (\ref{eq:cf1}).  The matrix inside the determinant operator can be decomposed into two parts.  
We observe that maximizing the expression below over $\Psi_S$
\begin{eqnarray*}
\arg \max_{\Psi_S}\frac{1}{2}\log\left|\left[\begin{array}{cc} \textbf{C}_{D} & 0\\0 & \textbf{C}_{R}+\hat{\textbf{C}}_R\end{array}\right] + \left[\begin{array}{cc} \textbf{D}_{SD}\Psi_S\textbf{D}_{SD}^{\dag} & \textbf{D}_{SD}\Psi_S\textbf{D}_{SR}^{\dag}\\\textbf{D}_{SR}\Psi_S\textbf{D}_{SD}^{\dag} & \textbf{D}_{SR}\Psi_S\textbf{D}_{SR}^{\dag}\end{array}\right]\right|\\
\end{eqnarray*}
is equivalent to maximizing,
\begin{eqnarray}
& & \arg \max_{\Psi_S} \frac{1}{2}\log\left|\left[\begin{array}{cc} \underbrace{ \textbf{D}_{SD}^{-1}\textbf{C}_{D}\textbf{D}_{SD}^{-1\dag}}_{{\textbf{C}}} & 0\\0 & \underbrace{\textbf{D}_{SR}^{-1}\left(\textbf{C}_{R}+\hat{\textbf{C}}_R\right)\textbf{D}_{SR}^{-1\dag}}_{{\textbf{D}}} \end{array}\right] + \left[\begin{array}{cc} \Psi_S & \Psi_S\\\Psi_S & \Psi_S\end{array}\right]\right| \nonumber \\
& \stackrel{(a)}{=} & \arg \max_{\Psi_S} \frac{1}{2}\log\left|\left[\begin{array}{cc} \Psi_S & \textbf{D}+\Psi_S\\\textbf{C}+\Psi_S & \Psi_S\end{array}\right]\right|
\nonumber\\
& \stackrel{(b)}{=} & \arg \max_{\Psi_S} \frac{1}{2}\log\left|\Psi_S\right|\left|\Psi_S -(\textbf{C}+\Psi_S)\Psi_S^{-1}(\textbf{D}+\Psi_S)\right| \nonumber\\
& = & \frac{1}{2}\log\left|\Psi_S\right|\left|\textbf{C}\Psi_S^{-1}\textbf{D}+\textbf{C}+\textbf{D}\right|. \label{eq:cf2}
\end{eqnarray}
Here, (a) follows since exchanging rows does not change the absolute value of the determinant and (b) follows from the properties of the determinant of a block matrix (see \cite{Horn}).
By Hadamard's inequality \cite{Cover_book}, the expression in Eqn.(\ref{eq:cf2}) is maximized when $\Psi_S$ is diagonal. 
 
\par Although we have shown that diagonal $\Psi_S$ maximizes the objective function, it remains to be shown that by choosing  $\Psi_S$ and $\Psi_R$ to be diagonal, the compression rate $\hat{\textbf{C}}_R$ at the relay is not altered. The compression rate is determined by the inequality constraint (\ref{eq:cf4}), which can be rewritten as,
\begin{eqnarray}\label{eq:cf_ineq}
\log \left|\hat{\textbf{C}}_R\right| & \geq & \log \left|\left[\begin{array}{cc} \textbf{D}_{SR}\Psi_S\textbf{D}_{SR}^{\dag}+\textbf{C}_R+\hat{\textbf{C}_R} & \textbf{D}_{SR}\Psi_S\textbf{D}_{SD}^{\dag}\\\textbf{D}_{SD}\Psi_S\textbf{D}_{SR}^{\dag} & \textbf{D}_{SD}\Psi_S\textbf{D}_{SD}^{\dag}+\textbf{C}_D\end{array}\right]\right|\nonumber\\
& & - \log \left|\textbf{D}_{SD}\Psi_S\textbf{D}_{SD}^{\dag}+\textbf{D}_{RD}\Psi_R\textbf{D}_{RD}^{\dag}+\textbf{C}_D\right|.\nonumber\\
& \stackrel{.}{=} & f(\Psi_S,\Psi_R)
\end{eqnarray}
Now we have to establish that diagonal $\Psi_S$ and $\Psi_R$ does not reduce the compression rate constraint set. We prove this by showing that for every choice of arbitrary $\Psi_S$ and $\Psi_R$, we can find a diagonal $\Psi_S^{*}$ and $\Psi_R^{*}$, which either gives the same constraint on the compression rate or improves it,.i.e.,
\begin{eqnarray*}
f(\Psi_S,\Psi_R) & \geq & f(\Psi_S^{*},\Psi_R^{*})
\end{eqnarray*}

First for any choice (diagonal or non-diagonal) of $\Psi_S$, we will find a condition on $\Psi_R$ which will maximally enlarge the compression rate constraint set. We consider the second term in the RHS of (\ref{eq:cf_ineq}), which is a function of $\Psi_R$, where $\Psi_R$ is a general non-negative definite matrix, not necessarily diagonal.

\begin{eqnarray}\label{eq:maxi_psir}
\log \left|\textbf{D}_{SD}\Psi_S\textbf{D}_{SD}^{\dag}+\textbf{D}_{RD}\Psi_R\textbf{D}_{RD}^{\dag}+\textbf{C}_D\right| & \equiv & \log \left|\Psi_S+\underbrace{\textbf{D}_{SD}^{-1}\textbf{D}_{RD}\Psi_R\textbf{D}_{RD}^{\dag}\textbf{D}_{SD}^{\dag -1}+\textbf{D}_{SD}^{-1}\textbf{C}_D\textbf{D}_{SD}^{\dag -1}}_{\equiv \textbf{C}_1}\right|\nonumber\\
& \stackrel{(a)}{=} & \log\left|\textbf{E}_S\Lambda_S\textbf{E}_S^{\dag}+\textbf{C}_1\right|\nonumber\\
& \stackrel{(b)}{=} & \log\left|\Lambda_S+\textbf{E}_S^{\dag}\textbf{C}_1\textbf{E}_S\right|\nonumber\\
& \stackrel{(c)}{\leq} & \log\left|\Lambda_S+\Lambda_{C_1}\right|
\end{eqnarray}
where in (a), $\Psi_S= \textbf{E}_S\Lambda_S\textbf{E}_S^{\dag}$ is the spectral decomposition of $\Psi_S$ with $\textbf{E}_S$ being the unitary eigenvector matrix associated with the diagonal eigenvalue matrix $\Lambda_S$. (b) is true because eigenvector matrix is unitary and (c) follows from the well known Hadamard's inequality (see {\em e.g.}\cite{Cover_book}), where $\Lambda_{C_1}$ is the eigenvalue matrix of $\textbf{C}_1$.

\par Since increase in the second term in the RHS of (\ref{eq:cf_ineq}) will make the constraint set larger, for a fixed choice of non-diagonal $\Psi_S$, equality is achieved in (\ref{eq:maxi_psir}) if $\Psi_R$ is chosen from a family of positive definite matrices (by varying $\Lambda_{C_1}$) diagonalized by $\textbf{E}_S$ and of the form $\Psi_R=\textbf{D}_{RD}^{-1}\textbf{D}_{SD}(\textbf{E}_S\Lambda_{C_1}\textbf{E}_S^{\dag}-\textbf{D}_{SD}^{-1}\textbf{C}_{D}\textbf{D}_{SD}^{\dag -1})\textbf{D}_{SD}^{\dag}\textbf{D}_{RD}^{\dag -1}$, satisfying the power constraint (\ref{eq:power_cons}) at the relay and maximizing the determinant in (\ref{eq:maxi_psir}). For diagonal $\Psi_S$, it is easy to verify that diagonal $\Psi_R$ will make $\textbf{C}_1$ diagonal.

\par Now with such choice of $\Psi_R$, we will show that it is sufficient to consider only diagonal $\Psi_S$. Let us consider the RHS of (\ref{eq:cf_ineq}) again,
\begin{eqnarray}\label{eq:maxi_psis}
\lefteqn{\log \left|\left[\begin{array}{cc} \textbf{D}_{SR}\Psi_S\textbf{D}_{SR}^{\dag}+\textbf{C}_R+\hat{\textbf{C}_R} & \textbf{D}_{SR}\Psi_S\textbf{D}_{SD}^{\dag}\\\textbf{D}_{SD}\Psi_S\textbf{D}_{SR}^{\dag} & \textbf{D}_{SD}\Psi_S\textbf{D}_{SD}^{\dag}+\textbf{C}_D\end{array}\right]\right|}\nonumber\\
& & - \log \left|\textbf{D}_{SD}\Psi_S\textbf{D}_{SD}^{\dag}+\textbf{D}_{RD}\Psi_R\textbf{D}_{RD}^{\dag}+\textbf{C}_D\right|\nonumber\\
& \equiv & \log \left|\left[\begin{array}{cc} \textbf{C}_R+\hat{\textbf{C}_R} & 0\\0 & \textbf{C}_D\end{array}\right]+\left[\begin{array}{cc} \textbf{D}_{SR} & 0\\0 & \textbf{D}_{SD}\end{array}\right]\left[\begin{array}{cc} \Psi_S & \Psi_S\\\Psi_S & \Psi_S\end{array}\right]\left[\begin{array}{cc} \textbf{D}_{SR}^{\dag} & 0\\0 & \textbf{D}_{SD}^{\dag}\end{array}\right]\right|\nonumber\\
& & - \log \left|\Psi_S+\textbf{C}_1\right|\nonumber\\
& \equiv & \log \left|\left[\begin{array}{cc} \textbf{C}_2 & 0\\0 & \textbf{C}_3\end{array}\right]+\left[\begin{array}{cc} \Psi_S & \Psi_S\\\Psi_S & \Psi_S\end{array}\right]\right| - \log \left|\Psi_S+\textbf{C}_1\right|\nonumber\\
& = & \log \left|\left[\begin{array}{cc} \Psi_S & \Psi_S+\textbf{C}_3\\\Psi_S+\textbf{C}_2 & \Psi_S\end{array}\right]\right|- \log \left|\Psi_S+\textbf{C}_1\right|\nonumber\\
& \stackrel{(a)}{=} & \log\left|\Psi_S\right|+\log \left|\left(\Psi_S+\textbf{C}_2\right)\Psi^{-1}_S\left(\Psi_S+\textbf{C}_3\right)-\Psi_S\right|- \log \left|\Psi_S+\textbf{C}_1\right|\nonumber\\
& = & \log\left|\Psi_S\right|+\log\left|\textbf{C}_2+\textbf{C}_3+\textbf{C}_2\Psi^{-1}_S\textbf{C}_3\right|- \log \left|\Psi_S+\textbf{C}_1\right|\nonumber\\
& \equiv & \log\left|\Psi_S\right|+\log\left|\Psi^{-1}_S+\textbf{C}_1^{'}\right|-\log\left|\Psi_S+\textbf{C}_1\right|\nonumber\\
& = & \underbrace{\log\left|\textbf{I}_n+ \Psi_S\textbf{C}_1^{'}\right|-\log\left|\Psi_S+\textbf{C}_1\right|}_{\equiv g(\Psi_S)}
\end{eqnarray}
where $\textbf{C}_2=\textbf{D}_{SR}^{\dag}\left(\textbf{C}_R+\hat{\textbf{C}_R}\right)\textbf{D}_{SR}, \textbf{C}_3=\textbf{D}_{SD}^{\dag}\textbf{C}_D\textbf{D}_{SD}$ and $\textbf{C}_1^{'}=\left(\textbf{C}_2+\textbf{C}_3\right)\textbf{C}^{-1}_2\textbf{C}^{-1}_3$ are all diagonal matrices and (a) follows from the fact that $\det\left(\begin{array}{cc} \textbf{A} & \textbf{B}\\\textbf{C} & \textbf{D}\end{array}\right)=\det\left(\textbf{A}\right)\det\left(\textbf{D}-\textbf{C}\textbf{A}^{-1}\textbf{B}\right)$, with $\textbf{A}=\textbf{D}=\Psi_S, \textbf{B}=\Psi_S+\textbf{C}_3$ and $\textbf{C}=\Psi_S+\textbf{C}_2$. The constraint set $\{\Psi_S:tr(\Psi_S)\leq P_S\}$ is closed and bounded and $g(\Psi_S)$ is continuous everywhere. Thus by the extreme value theorem (see \cite{Rudin}), $g(\Psi_S)$ has a maxima and mininma in the constraint set and the extreme points can be shown to be attained by the diagonal $\Psi_S$ by differentiating $g(\Psi_S)$ w.r.t. $\Psi_S$ and solving the KKT conditions (the proof is similar to the proof of Theorem 3 in \cite{Gamal} and thus omitted for brevity). So for every non-diagonal $\Psi_S$, we can find a diagonal $\Psi_S$ which gives the same value of $g(\Psi_S)$. Thus by choosing $\Psi_S$ and $\Psi_R$ to be diagonal, the compression rate is not reduced and a diagonal $\Psi_S$ maximizes the objective function Eqn. (\ref{eq:cf1}). This completes the proof.            
\end{proof}

\section{Proof of Optimality of CF Achievable Rate}
\begin{proof}
\par We re-examine the achievable rate for compress-and-forward strategy (see Eqn (\ref{eq:CF_rate})). If $\hat{y}_R^{nc}$ is a deterministic, invertible function of $y_R^{nc}$ then,
\begin{eqnarray*}
C_{nCF}^c & = & \sup I(x_S^n;y_D^{nc},\hat{y}_R^{nc}|x_R^n) = \sup I(x_S^n;y_D^{nc},{y}_R^{nc}|x_R^n),
\end{eqnarray*}
subject to the constraint,
\begin{eqnarray*}
I(x_R^n;y_D^{nc}) & \geq & I(y_R^{nc};\hat{y}_R^{nc}|x_R^n,y_D^{nc})\\
I(x_R^n;y_D^{nc}) & \stackrel{(a)}{\geq} & h(y_R^{nc}|x_R^n,y_D^{nc})\\
I(x_R^n;y_D^{nc}) & \geq & I(x_S^n;y_R^{nc}|x_R^n,y_D^{nc})\\
I(x_S^n;y_D^{nc}|x_R^n)+ I(x_R^n;y_D^{nc}) & \geq & I(x_S^n;y_D^{nc}|x_R^n)+I(x_S^n;y_R^{nc}|x_R^n,y_D^{nc})\\
I(x_S^n,x_R^n;y_D^{nc}) & \geq & I(x_S^n;y_R^{nc},y_D^{nc}|x_R^n).
\end{eqnarray*}
Here, (a) follows from the fact that $h(\hat{y}_R^{nc}|y_R^{nc},x_R^n,y_D^{nc})=0$ as $\hat{y}_R^{nc}$ is a deterministic, invertible function of $y_R^{nc}$. Under this condition it is clear that the cut-set upper bound is also given by,
$
C_n^{up}= \sup I(x_S^n;y_R^{nc},y_D^{nc}|x_R^n)
$ which in turn is equivalent to  $C_{nCF}^c$.
\end{proof}

\section{Proof of limiting DF Achievable Rate}
\begin{proof}
\par We need to show that,
\begin{eqnarray*}
\lim_{n\rightarrow\infty}C_{nDF}^c(P_S,P_R) & = & C_{DF}(P_S,P_R).
\end{eqnarray*}
To prove this we will use a theorem on minimax optimization \cite{Luenberger} stated below. Given, $\Phi:X\times Z\mapsto\mathcal{R}$, consider the following minimax optimization problem,
\begin{eqnarray*}
\min_{x\in X}\max_{z \in Z} \Phi(x,z).
\end{eqnarray*}
Define,
\begin{eqnarray*}
r_x(z)=-\Phi(x,z)\; \; \mbox{if} \; \; z\in Z, r(z)=\max_{x\in X} r_x(z),\\
t_z(x)=\Phi(x,z) \; \; \mbox{if} \; \; x\in X, t(x)=\max_{z\in Z} t_z(x).
\end{eqnarray*}
\begin{theorem}\cite{Luenberger}\label{thm:maximin}
If $X$ and $Z$ are convex and compact, $t_z(.)$ and $r_x(.)$ are closed and convex for each $z\in Z$ and $x\in X$, respectively, then
\begin{eqnarray*}
\min_{x\in X}\max_{z \in Z} \Phi(x,z) = \max_{z \in Z} \min_{x\in X} \Phi(x,z).
\end{eqnarray*}
\end{theorem}
Now,
\begin{eqnarray*}
C_{nDF}^c(P_S,P_R) & = & \max_{{\cal P}^D} \min \{C_{1nDF}^c,C_{2nDF}^c\}\\
& \stackrel{(a)}{=} & \max_{{\cal P}^D} \min_{0\leq\lambda\leq 1}(\lambda C_{1nDF}^c+\bar{\lambda}C_{2nDF}^c),\\
\end{eqnarray*}
where, $\bar{\lambda}=1-\lambda$ and here (a) follows from the fact that $\lambda C_{1nDF}^c+\bar{\lambda}C_{2nDF}^c$ is a line connecting $C_{1nDF}^c$ and $C_{2nDF}^c$, $\lambda=0$ and $\lambda=1$ are two extreme points of this line. Let,
\begin{eqnarray*}
U_n &\equiv & \left\{\alpha(\omega_i):0\leq\alpha(\omega_i)\leq1\right\}\\
V_n & \equiv & \left\{P_S(\omega_i):\frac{1}{n}\sum_{i=1}^{n}P_S(\omega_i)\leq P_S\right\}\\
W_n & \equiv & \left\{P_R(\omega_i):\frac{1}{n}\sum_{i=1}^{n}P_R(\omega_i)\leq P_R\right\}\\
Z_n & \equiv & \left\{z_n=(u_n, v_n, w_n):u_n \in U_n, v_n \in V_n, w_n \in W_n \right\}\\
X & \equiv & \{\lambda:0\leq\lambda\leq 1\}\\
\Phi(x,z_n) & = & \lambda C_{1n}^c+\bar{\lambda}C_{2n}^c.
\end{eqnarray*}
\par It is easy to see that both $X$ and $Z_n$ are convex sets and they are compact also since they are closed and bounded (assuming the peak power in a sub-band is bounded). And since $\Phi(x,z_n)$ is a continuous and bounded function (as capacity is bounded) of $(x,z_n)$, the solution of the problem exists for all $n$. So we can write,
\begin{eqnarray*}
\max_{z_n \in Z_n} \min_{x\in X} \Phi(x,z_n) & \equiv & \max_{u_n \in U_n, v_n \in V_n, w_n \in W_n} \min_{x\in X} \Phi(x,u_n, v_n, w_n)\\
& = & \max_{u_n \in U_n} \max_{v_n \in V_n} \max_{w_n \in W_n} \min_{x\in X} \Phi(x,u_n, v_n, w_n).
\end{eqnarray*}
For a fixed $w_n \in W_n$, $t_{w_n}(x)=\Phi(x,z_n)$ is a closed and convex function of $x$ and similarly, for a fixed $x\in X$, $r_x(w_n)=-\Phi(x,z_n)$ is a closed and convex function of $w_n$ (can be proved using the fact that if $f()$ and $g()$ are concave function of $x$ and $g()$ is nondecreasing, then $g(f(x))$ is a concave function). And hence by applying Theorem \ref{thm:maximin},
\begin{eqnarray*}
\max_{\{z_n\in Z_n\}} \min_{\{x\in X\}}(\lambda C_{1nDF}^c+\bar{\lambda}C_{2nDF}^c) & = & \max_{u_n \in U_n} \max_{v_n \in V_n} \min_{x\in X} \max_{w_n \in W_n} (\lambda C_{1nDF}^c+\bar{\lambda}C_{2nDF}^c).
\end{eqnarray*}
Now since $-\Phi(x,z_n)$ is also convex function of $u_n$ and $v_n$ for a fixed value of the other parameters, we can apply the same theorem twice to finally get,
\begin{eqnarray*}
\max_{\{z_n\in Z_n\}} \min_{\{x\in X\}}(\lambda C_{1nDF}^c+\bar{\lambda}C_{2nDF}^c) & = & \min_{\{x\in X\}} \max_{\{z\in Z_n\}}(\lambda C_{1nDF}^c+\bar{\lambda}C_{2nDF}^c).
\end{eqnarray*}
\par Now we show that,
\begin{eqnarray*}
\lim_{n\rightarrow\infty}\max_{Z_n} (\lambda C_{1nDF}^c+\bar{\lambda}C_{2nDF}^c) & = & \max_Z (\lambda C_{1DF}+\bar{\lambda}C_{2DF})\\
& = & \max_Z \Phi(x,z),
\end{eqnarray*}
where,
\begin{eqnarray*}
Z & \equiv & (u, v, w)\\
& = & \left\{(\alpha(\omega),P_S(\omega),P_R(\omega)):0\leq\alpha(\omega)\leq1, \frac{1}{2\pi}\int_{-\pi}^{\pi}P_S(\omega)\leq P_S, \frac{1}{2\pi}\int_{-\pi}^{\pi}P_R(\omega)\leq P_R\right\}.
\end{eqnarray*}
Since, $\Phi(x,z)$ is a strictly concave function in each of $u, v$ and $w$ and since each of $U, V$ and $W$ are convex constraint sets, $\Phi(x,z)$ achieves its maximum at some unique $z=z^*$. Similarly, $\Phi(x, z_n)$ attains its maxima at some unique $z_n=z_n^*$.
\par If a function is bounded and almost everywhere continuous on the interval $[-\pi, \pi]$ then it is
Riemann integrable on the interval,i.e.,
\begin{eqnarray*}
\lim_{n\rightarrow \infty} \frac{1}{n} \sum_{i=1}^n f(x_i) & = & \frac{1}{2\pi} \int_{-\pi}^{\pi} f(x)dx.
\end{eqnarray*}
Since the capacity of a power constrained system is finite, it can be shown that (for details see Lemma 3 of \cite{verdu}),
\begin{eqnarray*}
\lim_{n\rightarrow\infty} \Phi(x,z_n^*) & = & \Phi(x,z^*)\\
\lim_{n\rightarrow\infty}\max_{Z_n} (\lambda C_{1nDF}^c+\bar{\lambda}C_{2nDF}^c) & = & \max_Z (\lambda C_{1DF}+\bar{\lambda}C_{2DF}).
\end{eqnarray*}
\par As $Z$ is closed and convex, applying Theorem \ref{thm:maximin} again to exchange the min-max, we get our desired result.
\end{proof}

\section{Power Allocation for DF protocol}
\begin{proof}
\par Exchanging integration and infimum operation reduces the original problem into $n$ optimization problems. Consider one such optimization problem in the sub-band $i$ for the DF lower bound on the capacity.
\begin{eqnarray*}
\max_{P_{ti}} \min \left\{C \left(P(\omega_i)\right),C \left(a_{SRi}P_{S}^1(\omega_i)\right)\right\},
\end{eqnarray*}
where, $P_t(\omega_i)=P_S(\omega_i)+P_R(\omega_i)$ and $P_S(\omega_i)=P_S^1(\omega_i)+P_S^2(\omega_i)$. $P_S^1(\omega_i)$ is used by the source for transmission to the relay, $P_S^2(\omega_i)$ for transmission to the destination and $P_R(\omega_i)$ is the power allocated by the relay in the sub-band $\omega_i$.
\par For fixed $P_{2}(\omega_i)=P_{S}^2(\omega_i)+P_{R}(\omega_i)$, the rate $C \left(P(\omega_i))\right)$ is maximized when we set,
\begin{eqnarray*}
P_{S}^2(\omega_i) & = & \frac{a_{SDi}}{a_{SDi}+a_{RDi}}P_{2}(\omega_i)\\
P_{R}(\omega_i) & = & \frac{a_{SDi}}{a_{SDi}+a_{RDi}}P_{2}(\omega_i).
\end{eqnarray*}
We then get a simplified problem,
\begin{eqnarray*}
C_{i}^* &=& \max_{P_{t}(\omega_i)} \min \left\{C_{1}(\omega_i),C_{2}(\omega_i)\right\}\;\; \mbox{where,}\\
C_{1}(\omega_i) & = & C \left(a_{SRi}P_{S}^1(\omega_i)\right)\\
C_{2}(\omega_i) & = & C \left(a_{SDi}P_{S}^1(\omega_i)+(a_{SDi}+a_{RDi})P_{2}(\omega_i)\right).
\end{eqnarray*}
\begin{lemma}
If $a_{SRi}<a_{SDi}$ then $C_{1i}<C_{2i}$ $\forall P_{t}(\omega_i)$, otherwise for maximizing $P_{t}(\omega_i)$, $C_{1i}=C_{2i}$.
\end{lemma}
\begin{proof}
The first part of the lemma is obvious. For the second part, let $P_{ti}^*=P_{1i}^*+P_{Si}^{1*}$ maximizes the objective function and $C_{1i}^*<C_{2i}^*$. Then by choosing $P_{Si}^{1}=(1+\epsilon)P_{Si}^{1*}, P_{1i}=P_{1i}^*- \epsilon P_{Si}^{1*}$, we can increase $C_{1i}$ as it is a increasing function of $P_{Si}^1$ and decrease $C_{2i}$. So by continuity, at maximizing $P_{ti}^*$, $C_{1i}=C_{2i}$. We can use similar techniques for the case of $C_{2i}<C_{1i}$.
\end{proof}
If $a_{SRi}<a_{SDi}$ the sub-channel is non-degraded,and $C_{1i}<C_{2i}$. So, $C_i^*=\max C_{1i}$ and then solving the optimization problem we get, $C_i^*=C\left(a_{SRi}\left(\nu_t -\frac{1}{a_{SRi}}\right)^+\right)$. The corresponding power allocations are given by,
\begin{eqnarray*}
P_{t}^*(\omega_i) =\left(\nu_t-\frac{1}{a_{SRi}}\right)^+\\
P_{S}^{1*}(\omega_i)= P_{t}^*(\omega_i),P_{S}^{2*}(\omega_i) =P_{R}^*(\omega_i)=0;
\end{eqnarray*}
$P_{2}^*(\omega_i)$ is given no power as it has no effect on the achievable rate of the sub-band.
\par If $a_{SRi}\geq a_{SDi}$ the sub-channel is degraded, and the sub-channel rate is maximized when $C_{1i}^*=C_{2i}^*$,
\begin{eqnarray*}
a_{SRi}P_{S}^{1*}(\omega_i) & = & a_{SDi}P_{S}^{1*}(\omega_i)+(a_{SDi}+a_{RDi})P_{2}^*(\omega_i)\\
P_{2}^*(\omega_i) & = & \frac{a_{SRi}-a_{SDi}}{a_{SDi}+a_{RDi}}P_{S}^{1*}(\omega_i).
\end{eqnarray*}
Now, using the relation $P_{t}^*(\omega_i)=P_{S}^{1*}(\omega_i)+P_{2}^*(\omega_i)$ we get,
\begin{eqnarray*}
P_{t}^*(\omega_i) & = & \left(1+\frac{a_{SRi}-a_{SDi}}{a_{SDi}+a_{RDi}}\right)P_{S}^{1*}(\omega_i).
\end{eqnarray*}
The achievable rate is maximized when $P_{S}^{1*}(\omega_i)=\left(\nu_t -\frac{1}{a_{SRi}}\right)^+$ which implies that,
\begin{eqnarray*}
P_{2}^*(\omega_i) = \frac{a_{SRi}-a_{SDi}}{a_{SDi}+a_{RDi}}\left(\nu_t -\frac{1}{a_{SRi}}\right)^+\\
\Rightarrow P_{S}^{2*}(\omega_i) = a_{SDi}\frac{a_{SRi}-a_{SDi}}{(a_{SDi}+a_{RDi})^2}\left(\nu_t-\frac{1}{a_{SRi}}\right)^+\\
P_{R}^*(\omega_i) = a_{RDi}\frac{a_{SRi}-a_{SDi}}{(a_{SDi}+a_{RDi})^2}\left(\nu_t-\frac{1}{a_{SRi}}\right)^+,
\end{eqnarray*}
where $\nu_t$ is the power price chosen to satisfy the overall total input power constraint.
\end{proof}

\section{Proof of Theorem $\ref{thm:DF_asynch}$}
\begin{proof}
Before proving the theorem, we introduce two key lemmas.
\begin{lemma}\label{lem:rearrange}
\begin{eqnarray*}
\det\left(\textbf{I}_{2n} + \frac{1}{\sigma_D^2}E[x^nx^{nT}]\textbf{G}\right) & = & \det\left(\textbf{I}_{2n} + \frac{1}{\sigma_D^2}\left[\begin{array}{cc} \Sigma_R & \Sigma_{RS}\\\Sigma_{SR} & \Sigma_S\end{array}\right]\left[\begin{array}{cc} \textbf{I}_n & \textbf{S}\\\textbf{S}^T & \textbf{I}_n\end{array}\right]\right),
\end{eqnarray*}
where,
\begin{eqnarray*}
\textbf{S}^T & = & \rho_{RS}\textbf{I}_n + \rho_{SR}\left[\begin{array}{cccccc} 0 & 1 & 0 & & & \\0 & 0 & 1 & & & \\ & 0 & 0 & \vdots & 0 & \\ & & & & 1 & 0\\ & & & & 0 & 1\\1 & & & & 0 & 0\end{array}\right],\\
\textbf{G} & = & \left[\begin{array}{cccccccc} 1 & \rho_{RS} &  &  &  &  &  &\rho_{SR}\\\rho_{RS} & 1 & \rho_{SR} & & & & & \\ & \rho_{SR} & 1 & \rho_{RS} & & & &\ddots\\ & & & & & \rho_{SR} & 1 & \rho_{RS}\\\rho_{SR} & & & & & &\rho_{RS} & 1\end{array}\right].
\end{eqnarray*}
\end{lemma}
This Lemma can be easily derived from Lemma 1 of \cite{verdu}. The key to showing Theorem $\ref{thm:DF_asynch}$ is proving that
the DFT decomposition is again optimal for DF. To this end, we will need the following Lemma.
\begin{lemma}\label{lem:diagonal}
\begin{eqnarray*}
\lefteqn{\det\left(\textbf{I}_{2n} + \frac{1}{\sigma_D^2}\left[\begin{array}{cc} \Psi_R & \Psi_{RS}\\\Psi_{SR} & \Psi_S\end{array}\right]\left[\begin{array}{cc} \textbf{I}_n & \textbf{D}\\\textbf{D}^* & \textbf{I}_n\end{array}\right]\right)}\\
& \leq & \prod_{i=1}^n\left\{\det\left(\textbf{I}_2+\frac{1}{\sigma_D^2}\left[\begin{array}{cc} \psi_{Rii} & \psi_{RSii}\\\psi_{SRii} & \psi_{Sii}\end{array}\right]\left[\begin{array}{cc} 1 & d_{ii}\\d_{ii}^* & 1\end{array}\right]\right)\right\},
\end{eqnarray*}
where, $\textbf{D}$ is a diagonal matrix with $d_{ii}$ as it's $i$-th diagonal element.
\end{lemma}
\par The proof of this Lemma uses appropriate permutation matrix $\textbf{P}$ to block-diagonalize the LHS, similar to the proof of Lemma $\ref{lem:opt_CF}$. With Lemma $\ref{lem:rearrange}$ and $\ref{lem:diagonal}$ in hand, we can now prove the Theorem $\ref{thm:DF_asynch}$.

\par The input-output relation of a discrete-time Gaussian circular MIMO relay channel model is given by,
\begin{eqnarray}
{y}_R^c(i) & = & x_S(i)+{n}_R^c(i),\\
\bar{y}_D^c(i) & = & \left[\begin{array}{c} {y}_{DR}^c(i)\\{y}_{DS}^c(i)\end{array}\right] = \left[\begin{array}{cc} 0 & \rho_{SR}\\0 & 0\end{array}\right]\left[\begin{array}{c} x_R(i-1)_n\\x_S(i-1)_n\end{array}\right] + \left[\begin{array}{cc} 1 & \rho_{RS}\\\rho_{RS} & 1\end{array}\right]\left[\begin{array}{c} x_R(i)\\x_S(i)\end{array}\right]\nonumber\\
& & + \left[\begin{array}{cc} 0 & 0\\\rho_{SR} & 0\end{array}\right]\left[\begin{array}{c} x_R(i+1)_n\\x_S(i+1)_n\end{array}\right]+ \left[\begin{array}{c} {n}_{DR}^c(i)\\{n}_{DS}^c(i)\end{array}\right],          
\end{eqnarray}
for $1\leq i \leq n$; where $((i)_n)$ equals $i$ modulo $n$ except when s is zero or an integer multiple of $n$, in which case $((i)_n)=n$. The noise processes ovar each $n$-block $\{{n}_{R}^c(i)\}$ and $\left\{\left[\begin{array}{cc} {n}_{DR}^c(i) & {n}_{DS}^c(i)\end{array}\right]\right\}$ are circular and its autocorrelation is a periodic repitition of the autocorrelation of the original noise samples within an $n$-block as defined in \cite{Andrea01}. The output of the channel in vector form can be written as,
\begin{eqnarray}\label{eq:vector_relay}
\left[\begin{array}{c} {y}_R^c(1) \\{y}_R^c(2) \\\vdots \\{y}_R^c(n-1) \\{y}_R^c(n)\end{array}\right] & = & \textbf{I}_n\left[\begin{array}{c} {x}_S(1) \\x_S(2) \\\vdots \\x_S(n-1) \\x_S(n)\end{array}\right]+ \left[\begin{array}{c} {n}_R^c(1) \\{n}_R^c(2) \\\vdots \\{n}_R^c(n-1) \\{n}_R^c(n)\end{array}\right],
\end{eqnarray}
\begin{eqnarray}\label{eq:vector_dest}
\left[\begin{array}{c} {y}_{DR}^c(1) \\{y}_{DS}^c(1) \\{y}_{DR}^c(2) \\\vdots \\{y}_{DR}^c(n) \\{y}_{DS}^c(n)\end{array}\right] & = & \textbf{G}\left[\begin{array}{c} x_R(1) \\x_S(1) \\x_R(2) \\\vdots \\x_R(n) \\x_S(n)\end{array}\right] + \left[\begin{array}{c} {n}_{DR}^c(1) \\{n}_{DS}^c(1) \\{n}_{DR}^c(2) \\\vdots \\{n}_{DR}^c(n) \\{n}_{DS}^c(n)\end{array}\right],
\end{eqnarray}
where, the Gaussian noise processes $\{{n}_{R}^c(i)\}$ and $\left\{\left[\begin{array}{cc} {n}_{DR}^c(i) & {n}_{DS}^c(i)\end{array}\right]\right\}$ have covariance matrix $\sigma_R^2\textbf{I}_n$ and $\sigma_D^2\textbf{G}$, respectively. So,
\begin{eqnarray}\label{eq:mac_simple}
I(x_S^n,x_R^n;\bar{y}_D^{nc}) & = & h(\bar{y}_D^{nc}) - h(\bar{y}_D^{nc}|x_S^n,x_R^n)\nonumber\\
& = & h(\textbf{G}x^n+\bar{n}_D^{nc}) - h(\textbf{G}x^n+\bar{n}_D^{nc}|x_S^n,x_R^n)\nonumber\\
& \stackrel{(a)}{=} & h(\textbf{G}x^n+\bar{n}_D^{nc}) - h(\bar{n}_D^{nc})\nonumber\\
& \stackrel{(b)}{\leq} & \frac{1}{2} \log\det\left(cov(\textbf{G}x^n+\bar{n}_D^{nc})\right) - \frac{1}{2} \log\det\left(\sigma_D^2\textbf{G}\right)\nonumber\\
& = & \frac{1}{2} \log\det\left(\textbf{I}_{2n} + \frac{1}{\sigma_D^2}E[x^nx^{nT}]\textbf{G}\right)\nonumber\\
& \stackrel{(c)}{=} & \det\left(\textbf{I}_{2n} + \frac{1}{\sigma_D^2}\left[\begin{array}{cc} \Sigma_R & \Sigma_{RS}\\\Sigma_{SR} & \Sigma_S\end{array}\right]\left[\begin{array}{cc} \textbf{I}_n & \textbf{S}\\\textbf{S}^T & \textbf{I}_n\end{array}\right]\right)\nonumber\\
& \stackrel{(d)}{=} & \frac{1}{2} \log\det\left(\textbf{I}_{2n} + \frac{1}{\sigma_D^2}\left[\begin{array}{cc} \Psi_R & \Psi_{RS}\\\Psi_{SR} & \Psi_S\end{array}\right]\left[\begin{array}{cc} \textbf{I}_n & \textbf{D}\\\textbf{D}^* & \textbf{I}_n\end{array}\right]\right)\nonumber\\
& \stackrel{(e)}{\leq} & \prod_{i=1}^n\left\{\det\left(\textbf{I}_2+\frac{1}{\sigma_D^2}\left[\begin{array}{cc} \psi_{Rii} & \psi_{RSii}\\\psi_{SRii} & \psi_{Sii}\end{array}\right]\left[\begin{array}{cc} 1 & d_{ii}\\d_{ii}^* & 1\end{array}\right]\right)\right\}        
\end{eqnarray}
where, $x^n=\{x_R(i),x_S(i)\}_{i=1}^n$ and (a) is justified since noise is independent of the inputs, (b) follows as Gaussian input maximizes entropy for a given input covariance matrix, (c) follows from the fact that $\textbf{S}$ is a circulant matrix and hence it gets diagonalized by the DFT matrix $\textbf{F}$ i.e., $\textbf{S}=\textbf{F}\textbf{D}\textbf{F}^\dag$, (d) follows from the Lemma $\ref{lem:rearrange}$ and (e) is true because of Lemma $\ref{lem:diagonal}$. 

\par Here, $d_{ii}$ is the $i$-th eigenvalue of the circulant matrix $\textbf{S}$ and is given by the DFT of the first column of $\textbf{S}$ and thus $d_{ii}=\rho_{RS}+\rho_{SR}e^{-j\omega_i}=\rho(\omega_i)$. Let, $\psi_{Rii}=P_R(\omega_i), \psi_{Sii}=P_S(\omega_i)$ be the are the power allocated by the source and the relay for $i$-th component of the channel and $\alpha(\omega_i)$ be the correlation between $X_{Si}$ and $X_{Ri}$ as defined in \cite{Cover79}, then $\psi_{RSii} = \sqrt{(1-\alpha(\omega_i))P_S(\omega_i)P_R(\omega_i)}$ and subtituting these values in ($\ref{eq:mac_simple}$) we get,
\begin{eqnarray}\label{eq:mac_memory}
I(x_S^n,x_R^n;\bar{y}_D^{nc}) & \leq & \frac{1}{2} \sum_{i=1}^n \log\left(1 + \frac{1}{\sigma_D^2}(P_S(\omega_i)+P_R(\omega_i)+2\sqrt{(1-\alpha(\omega_i))P_S(\omega_i)P_R(\omega_i)}\rho(\omega_i))\right.\nonumber\\
& & \left. +\frac{1}{\sigma_D^4}\alpha(\omega)P_S(\omega_i)P_R(\omega_i)(1-\rho^2(\omega_i))\right)
\end{eqnarray}
Similarly,
\begin{eqnarray}\label{eq:broad_memory}
I(x_S^n;{y}_R^{nc}|x_R^n) & = & h(x_c^n + {n}_R^{nc}) - h(\bar{n}_R^{nc})\nonumber\\
& \leq & \frac{1}{2} \log\det\left(\textbf{I}_n + \frac{1}{\sigma_R^2}\Sigma\right)\nonumber\\
& = & \frac{1}{2} \log\det\left(\textbf{I}_n + \frac{1}{\sigma_R^2}\Psi\right)\nonumber\\
& \stackrel{(a)}{\leq} & \frac{1}{2} \sum_{i=1}^n \log\left(1 + \frac{1}{\sigma_R^2}\alpha(\omega_i)P_S(\omega_i)\right) 
\end{eqnarray}
where, $x_c^n=x_S^n|x_R^n$ and $\Sigma$ is the conditional covariance matrix and it can be expressed in terms of the input covariance matrices and its given by $\Sigma=E[x_c^n(x_c^n)^T]=\Sigma_S - \Sigma_{SR}\Sigma_{R}^{-1}\Sigma_{SR}^{\dag}$. (a) can be shown using Lemma $\ref{lemma:diag}$. Now combining Eqn. ($\ref{eq:mac_memory}$) and ($\ref{eq:broad_memory}$) and taking the limit in the block-length $n$, we get the expression of Theorem $\ref{thm:DF_asynch}$.
\end{proof}

\end{document}